\documentclass[11pt,letterpaper]{article}

\pdfoutput=1

\usepackage[margin=1in]{geometry}

\usepackage[
  bookmarks=true,
  bookmarksnumbered=true,
  bookmarksopen=true,
  pdfborder={0 0 0},
  breaklinks=true,
  colorlinks=true,
  linkcolor=black,
  citecolor=black,
  filecolor=black,
  urlcolor=black,
]{hyperref}

\usepackage{amsmath}
\usepackage{amsthm}
\usepackage{amssymb}
\usepackage{mathtools}
\usepackage[capitalise]{cleveref}
\usepackage{dsfont}
\usepackage{nicefrac}
\usepackage{multirow}

\usepackage{amsfonts}
\usepackage{url}
\usepackage{verbatim}
\usepackage{bbm}
\allowdisplaybreaks

\usepackage[square]{natbib}

\newcommand{\mB}{{m_B}}

\newcommand{\mS}{{m_S}}

\newcommand{\vecb}{{\vec{b}}}

\newcommand{\vecs}{{\vec{s}}}

\newcommand{\vecq}{{\vec{q}}}

\newcommand{\E}{\mathbb{E}}
\newcommand{\vareps}{\varepsilon}

\newcommand\ZZZ[3]{\ensuremath{z^{({#1})}_{({#2},{#3})}}}
	
\newcommand\X[3]{\ensuremath{x^{({#1})}_{({#2},{#3})}}}

\newcommand{\cout}[1]{}

\def\btr{\ensuremath{\mathrm{BTR}}}
\def\str{\ensuremath{\mathrm{STR}}}
\def\median{\ensuremath{\mathrm{MEDIAN}}}
\def\opt{\ensuremath{\mathrm{OPT}}}
\mathchardef\mhyphen="2D 
\def\sbopt{\ensuremath{\mathrm{FEASIBLE\mhyphen OPT}}}

\Crefname{equation}{Equation}{Equations}

\newtheorem{theorem}{Theorem}[section]
\newtheorem{lemma}[theorem]{Lemma}
\newtheorem{proposition}[theorem]{Proposition}
\newtheorem*{open-problem}{Open Problem}

\def\citeN{\citet}
\def\cite{\citep}

\hypersetup{
  pdfauthor      = {Moshe Babaioff <moshe@microsoft.com>, Kira Goldner <kgoldner@cs.columbia.edu>, and Yannai A. Gonczarowski <yannai@gonch.name>},
  pdftitle       = {Bulow-Klemperer-Style Results for Welfare Maximization in Two-Sided Markets},
}

\begin{document}

\begin{titlepage}

\title{Bulow-Klemperer-Style Results for\\Welfare Maximization in Two-Sided Markets}
\author{Moshe Babaioff%
\thanks{%
    Microsoft Research (\url{moshe@microsoft.com})}
\and Kira Goldner%
\thanks{%
    Columbia University (\url{kgoldner@cs.columbia.edu}). Work carried out in part while at the University of Washington. Supported by NSF CCF-1813135 and by a Microsoft Research PhD Fellowship. Partially supported by the European Research Council (ERC) under the European Union's Horizon 2020 research and innovation programme (grant agreement No 740282).}
\and Yannai A. Gonczarowski%
\thanks{%
        Microsoft Research (\url{yannai@gonch.name}). Work carried out in part while at the Hebrew University of Jerusalem and at Tel Aviv University.
        Supported by the Adams Fellowship Program of the Israel Academy of Sciences and Humanities; by ISF grants 1435/14, 317/17, and 1841/14 administered by the Israeli Academy of Sciences; by Israel-USA Bi-national Science Foundation (BSF) grant number 2014389; and by the European Research Council (ERC) under the European Union's Horizon 2020 research and innovation programme (grant agreement No 740282), and under the European Union's Seventh Framework Programme (FP7/2007-2013) / ERC grant agreement number 337122.
        }
}
\date{December 23, 2019}

\maketitle

\pagenumbering{roman}
\thispagestyle{empty}

\begin{abstract}
We consider the problem of welfare (and gains-from-trade) maximization in two-sided markets using simple mechanisms that are prior-independent. The seminal impossibility result of \citeN{MyersonS83} shows that even for \emph{bilateral trade}, there is no \emph{feasible} (individually rational, truthful, and budget balanced) mechanism that has welfare as high as the optimal-yet-infeasible VCG mechanism, which attains maximal welfare but runs a deficit. On the other hand, the optimal feasible mechanism needs to be carefully tailored to the Bayesian prior, and even worse, it is known to be extremely complex, eluding a precise description.

In this paper we present Bulow-Klemperer-style results to circumvent these hurdles in double-auction market settings. We suggest using the \emph{Buyer Trade Reduction (BTR)} mechanism, a variant of \citeauthor{McAfee92}'s mechanism, which is feasible and simple (in particular, it is deterministic, truthful, prior-independent, and anonymous). First, in the setting in which the values of the buyers and of the sellers are sampled independently and identically from the same distribution, we show that for any such market of any size, BTR with one additional buyer whose value is sampled from the same distribution has expected welfare at least as high as the optimal-yet-infeasible VCG mechanism in the original market.

We then move to a more general setting in which the values of the buyers are sampled from one distribution, and those of the sellers from another, focusing on the case where the buyers' distribution first-order stochastically dominates the sellers' distribution. We present both upper bounds and lower bounds on the number of buyers that, when added, guarantees that BTR in the augmented market have welfare at least as high as the optimal in the original market. Our lower bounds extend to a large class of mechanisms, and all of our positive and negative results extend to adding sellers instead of buyers. In addition, we present positive results about the usefulness of pricing at a sample for welfare maximization (and more precisely, for gains-from-trade approximation) in two-sided markets under the above two settings, which to the best of our knowledge are the first sampling results in this context.

\end{abstract}

\end{titlepage}

\pagenumbering{arabic}

\section{Introduction}

The field of Mechanism Design aims to design mechanisms to 
obtain a certain social objective
in settings in which participants act strategically.  
At the heart of Algorithmic Mechanism Design lie the questions of (1) how simple can such a designed mechanism be, and (2) how much does such a mechanism have to know about the participants. For single-item auctions, for instance, the seminal work of \citeN{Myerson81} completely characterizes the revenue-maximizing truthful mechanism given a Bayesian prior over the buyers' valuations of the item. This auction turns out to depend on the Bayesian prior even when the same prior applies for all buyers. 
Nonetheless, the seminal work of \citeN{BK} (henceforth BK) remarkably shows that in this case, 
instead of opting to use \citeauthor{Myerson81}'s revenue-maximizing mechanism that is tailored to the prior, 
the seller can use an extremely simple mechanism that requires no prior knowledge whatsoever --- a second-price auction with no reserve --- as long as the seller can recruit even just one more potential buyer from the same population (i.e., for which the same prior applies) to join in on the bidding. As BK show, the seller suffers \emph{no loss at all} in revenue compared to using \citeauthor{Myerson81}'s revenue-maximizing mechanism with the original buyers. This result is remarkable in the simplicity and robustness of the mechanism, in the minuscule amount of buyers that the seller needs to recruit, and in guaranteeing not only an approximation, but in fact no loss \emph{whatsoever}.

\paragraph{Two-sided markets:} Mechanisms for two-sided auction markets, where both sellers
and buyers
are \emph{strategic participants}, are notoriously hard to design even for the goal of welfare maximization. This is because, in addition to standard participation and truthfulness properties, two-sided auction mechanisms must maintain \emph{budget balance}: the mechanism may not pay sellers more than it charges the buyers.
This added requirement rules out the standard recipe for welfare maximization, that is, the VCG mechanism.
In fact, the seminal impossibility result of 
\citeN{MyersonS83} shows that even for \emph{bilateral trade}, that is, even for the setting where a single seller wishes to sell a single item to a single buyer, there is no \emph{feasible} (individually rational, truthful and budget balanced) mechanism that is as efficient as the VCG mechanism, which attains optimal welfare yet runs a deficit. 
Given this result, one may try to turn to the optimal \emph{feasible} mechanism --- the mechanism with the highest expected welfare among all feasible mechanisms. Unfortunately, that mechanism needs to be carefully tailored to the Bayesian prior --- far more carefully than \citeauthor{Myerson81}'s mechanism in fact --- and even worse, it is known to be extremely complex \cite{MyersonS83}, eluding a precise description.

This hardship of course carries over from bilateral trade to more general two-sided auction markets, and in particular to the double-auction setting we study in this paper: there are a number of sellers and a number of buyers, each seller holds one item, with all items interchangeable, and each buyer is interested in one item.
Both the intimate dependence of the optimal feasible mechanism on the Bayesian prior over the buyers' valuations and the Bayesian prior over the sellers' valuations (similarly to the revenue-maximization setting described above), and its overwhelming complexity (far beyond the revenue-maximization setting described above) lead us to ask whether a result similar in spirit to that of BK can be proven in this context: 
Is there a simple \emph{robust} (prior-independent and feasible) mechanism that can guarantee no loss whatsoever (with respect to the optimal feasible mechanism? with respect to the optimal-yet-infeasible VCG mechanism?) if only one more buyer can be recruited into the market?\footnote{We focus on recruiting buyers rather than sellers both because (1) this intuitively seems like a more natural task as in many markets buyers are more abundant, and (2) because this does not change the initial pre-trade welfare in the market, which means that the questions of guaranteeing the same (post-trade) social welfare or the same \emph{gains from trade} (difference between pre-trade and post-trade social welfare) coincide. Nevertheless, our results can be translated into analogous results for adding sellers rather than buyers. See \cref{app:sellers-buyers} for more details.}\textsuperscript{,}\footnote{An alternative approach may have been to search for a simple robust mechanism that in the original one-seller-one-buyer market, with no additional buyers, at least reasonably approximates the expected optimal welfare. It is well known, though, that in the bilateral trade setting no such mechanism can guarantee any approximation. }

\paragraph{A simple, robust mechanism:} To make the question we have just posed better defined, one must ask oneself what qualifies as a simple mechanism. Surely, the second-price auction of BK is simple --- what is the analogue of their mechanism for two-sided markets?\footnote{Recall that the VCG mechanism is not budget balanced --- it runs a deficit in two-sided markets, so it is not feasible.}
Even more fundamentally, what is the analogue of their mechanism for a two-sided market with only one seller and many buyers? To understand this, we note that BK assume that the seller values her item at some a fixed valuation that is less than any buyer's valuation of the item, so the seller always recovers her initial valuation for the item. Such an assumption in two-sided markets would make welfare maximization trivial, and even worse --- would completely throw away one of the key aspects of two-sided markets: the \emph{a priori} uncertainty regarding who should get the item in the welfare-maximizing outcome.

An immediate generalization of BK's mechanism that comes to mind for the case where the seller's valuation for the item may be lower than any buyer's valuation is to set the seller's valuation as a \emph{reserve price} for the second-price auction. This mechanism is budget balanced, and its expected welfare, if the seller is truthful, is in fact optimal like that of the optimal-yet-infeasible VCG mechanism. However, the problem (unsurprisingly, given \citeauthor{MyersonS83}'s impossibility) is that this mechanism is not truthful for the seller: in many cases, the seller will have a clear incentive to misrepresent her valuation to be higher than it really is. Therefore, this mechanism is not feasible.

Fortunately, it is not hard to see that this mechanism can be slightly
tweaked into a truthful mechanism. Instead of a second-price auction that starts at the seller's valuation as the reserve price, we can use what we call a \emph{second-price auction with seller veto}: first come up with a price using a second-price auction with no reserve price among the buyers, however (unlike BK's second-price auction) only perform the trade at this price if it is no less than the seller's value (so the seller in effect vetoes trades that would give her negative utility).
While maybe less natural at first glance, this mechanism shares many properties with a second-price auction with the seller's true valuation as the reserve price: it guarantees to the seller the exact same utility,
it also generalizes BK's mechanism (that is, their mechanism coincides with this one when the seller's valuation is less than any buyer's valuation), it is budget balanced, and it is prior-independent. Furthermore, unlike a second-price auction that has the seller's reported valuation as the reserve price, this mechanism is truthful and is therefore feasible and robust. By the impossibility result of \citeN{MyersonS83}, the feasibility of this mechanism comes at a price of course: second-price auction with seller veto does not trade in certain cases in which a second-price auction with the seller's true valuation as a reserve price would have traded, and therefore its welfare is lower in such cases.\footnote{This indeed happens when the price coming out of the second-price auction is lower than the seller's value, while the highest value buyer has higher value than the seller.}

\paragraph{A generalized mechanism:} Given ``second-price auction with seller veto'' as a candidate for a simple and robust mechanism for the setting of one seller and many buyers (we have still not given any evidence that it performs well, though!), how does one generalize this mechanism to the more general double-auction market
setting with multiple
sellers? Fortunately, it turns out that an established robust mechanism can be tweaked to do precisely this: the \emph{Buyer Trade Reduction} (henceforth \emph{BTR}) mechanism is inspired by the celebrated mechanism of \citeN{McAfee92}, and it is a generalization of the second-price auction with seller veto to double-auction markets.
Much like the classic mechanism of \citeN{McAfee92}, BTR is also deterministic, robust (prior-independent, individually rational, truthful, budget-balanced), and \emph{anonymous} (essentially, all sellers are treated the same and all buyers are treated the same, up to tie breaking). \citeauthor{McAfee92}'s mechanism has received attention for its good (approximate) welfare properties for, roughly speaking, large markets in which the welfare-maximizing trade size is typically large. These nice properties carry over to BTR as well. 

The Buyer Trade Reduction (BTR) mechanism proceeds as follows: it first finds the welfare-maximizing trade by pairing the highest-value buyer with the lowest-value seller, the second-highest-value buyer with the second-lowest-value seller, and so on,\ as long as such pairs can be formed where the buyer's value is at least the seller's value. Trading in these pairs would have indeed been optimal (this is what the VCG mechanism would have done, but while doing so VCG would have run a deficit), however to be truthful and budget balanced, the mechanism suggests a price at which to conduct these trades. BTR takes the value of the next buyer in line and offers it as the price to all pairs.\footnote{This is where BTR differs from the original mechanism of \citeN{McAfee92}, which proposes as price the \emph{average} of the value of the next buyer in line and the value of the next seller in line.
} If all pairs accept, this price is chosen and all these pairs trade (resulting in a welfare-maximizing trade). Otherwise, the last pair formed (with the smallest gains) is \emph{reduced}: it does not trade, and all other pairs trade with the trading buyers paying the value of the reduced buyer and the trading sellers receiving the value of the reduced seller. This mechanism indeed generalizes a second-price auction with seller veto: in case of one seller, the seller is paired with the highest-value buyer (unless the seller's value is higher than any buyer's, 
in which case it is optimal to have no trade), 
and this pair trades if both accept the value of the next buyer in line  --- this value is indeed exactly what would have been the price determined by a second-price auction with no reserve among the buyers.

Now that we have our candidate for a simple, robust mechanism in hand, we need to understand how good it is, both in absolute terms (whether it does or does not allow for a BK-type result) and in relative terms (whether some other simple, robust mechanism could conceivably outperform it). In this paper we give positive results to questions of both of these types.

\renewcommand{\floatpagefraction}{.9}

\subsection{Our Results}

\begin{table*}
	\centering
	\begin{tabular}{|c|c|c|c|c|c|c|c|}
	\hline
	\multicolumn{2}{|c|}{\textbf{Setting}} &
	\multicolumn{2}{c|}{\textbf{Sufficient \#buyers to add}} &
	\multicolumn{2}{c|}{\textbf{Insufficient \#buyers to add}} \\ 
	(\#S, \#B)& Condition& Bound & Theorem & Bound & Theorem \\
    \hline
    $\mS,\mB$ & i.i.d.\ ($F_B=F_S$) & 1 & Thm~\ref{thm:iid} & 0 & MS [\citeyear{MyersonS83}]  \\
  	\hline
	$\mS,\mB$ & arbitrary $F_B, F_S$ & impossible, by $\Rightarrow$ & & any finite number & Prop~\ref{obs-all-mech-reg-not-suff}
	\\
    \hline
    1,1 & $F_B$ FSD $F_S$ & 4  & Thm~\ref{prop:sd-1-1-ub} & 1 & Thm~\ref{thm:general-mech-lb-12} \\
    \hline
    $1, \mB$ & $F_B$ FSD $F_S$ & $4 \sqrt{\mB}$ & Prop~\ref{prop:sd-1-r-ub} & $\lfloor\log_2 \mB\rfloor$ & Thm~\ref{thm:sd-1-r-lb} \\
    \hline
   $\mS, \mB$ & $F_B$ FSD $F_S$ & $\mS(\mB + 4 \sqrt{\mB})$ & Thm~\ref{thm:sd-m-r-ub} & $\Uparrow$ & \\
    \hline
  \end{tabular}\caption{Overview of our main Bulow-Klemperer-style results. For each setting with the listed original number of sellers and buyers and the given conditions on the seller and buyer distributions, we list our upper and lower bounds on the number of buyers that, when added, guarantee welfare (and gains from trade) at least as high as the optimum of the original market.
The left ``Bound'' column states a number such that for any distributions satisfying the condition, when adding this number of additional buyers, the welfare (and gains-from-trade) of BTR with the added buyers is at least the optimum of the original market. The right ``Bound'' column states a number such that for any anonymous robust deterministic mechanism, there exist distributions satisfying the condition such that even when adding this number of additional buyers, the welfare (and gains-from-trade) of the mechanism with the added buyers is strictly less than the optimum
of the original market.}\label{table:bkresults}
\end{table*}

\begin{table*}
  \centering
  \begin{tabular}{|c|c|c|c|}
  \hline
    \multirow{2}{*}{\textbf{Condition}} & \textbf{Guaranteed} & \textbf{Approximation that some} & \multirow{2}{*}{\textbf{Theorem}} \\
    & \textbf{approximation} & \textbf{distribution does not attain} & \\
    \hline
    i.i.d.\ ($F_B=F_S$) & $1/2$ & $>1/2$ & Thm~\ref{cor:iid-sample-half} \\
    \hline
    $F_B$ FSD $F_S$ & $1/4$ & $>7/16$ & Thm~\ref{thm:FSD-11-sample}; Prop~\ref{prop:FSD-11-sample-nohalf} \\
    \hline
  \end{tabular}\caption{Overview of our results for pricing at a sample. For each setting with the given condition on the seller and buyer distributions, we list our upper and lower bounds on the guaranteed fraction of the optimum gains-from-trade obtained by pricing at a fresh sample drawn from the buyer's distribution.}\label{table:sampleresults}
\end{table*}

\paragraph{Identically distributed values for sellers and buyers:} As discussed above, in two-sided markets the optimal welfare is not attainable by any feasible mechanism, and furthermore, the optimal feasible mechanism is extremely complex and intimately prior-dependent. Our first main result shows that despite these challenges, in the case where all buyer values and seller values are drawn from the same distribution, a strong analogue of the result of BK holds for welfare in two-sided markets rather than revenue in auctions. That is, regardless of this distribution and regardless of the initial number of buyers or of the initial number of sellers,\footnote{As with the result of BK, though, of particular interest are small markets. This is because for large markets, the simple mechanisms considered by BK and by us perform fairly well on the original market, even without adding any buyers.} recruiting even one extra buyer and using BTR\footnote{As noted above, all our results can be recast as results for adding sellers rather than buyers. The appropriate simple mechanism in this case is the Seller Trade Reduction (STR) mechanism, an anonymous robust deterministic mechanism that is identical to BTR except that the trade price is taken to be the value of the next seller in line rather the value of the next buyer in line. See \cref{app:sellers-buyers} for more details.} gives higher expected social welfare (equivalently, higher expected gains from trade) than using not only the optimal feasible mechanism, but in fact the optimal-yet-infeasible VCG mechanism with the original buyers and sellers:

\begin{theorem}[See \cref{thm:iid}]\label{intro-iid}
	For any number of sellers $\mS$ and any number of buyers $\mB$, if all seller and buyer values are sampled i.i.d.,\ it holds that\footnote{Here and henceforth, for any truthful mechanism $M$, we write $M(\mS,\mB)$ to mean the expected gains from trade (difference between post-trade welfare and pre-trade welfare) of $M$ when there are $\mS$ sellers with values sampled i.i.d.\ from the sellers' value distribution, and $\mB$ buyers with values sampled i.i.d.\ from the buyers' value distribution. (In \cref{intro-iid}, these distributions coincide.) We use $\opt$ to denote the optimal-yet-infeasible VCG mechanism, $\sbopt$ to denote the optimal feasible mechanism, and $\btr$ to denote the Buyer Trade Reduction mechanism.} \[\btr(\mS,\mB+1)\geq \opt(\mS,\mB).\]
\end{theorem} 
We stress that the above result places no assumptions whatsoever on the distribution $F$, and in particular, does \emph{not} assume the regularity condition that the celebrated result of BK imposes (their result indeed does not hold for some distributions that are not regular, while our result does). While the results are analogous (for different settings and different objectives, of course), the proofs of the two results are very different and do not seem to be related.

\paragraph{Sellers and buyers of two different populations:} Can \cref{intro-iid} be extended beyond the i.i.d.\ case, to a setting where the values of the sellers  are drawn from one distribution, $F_S$, and the values of the buyers are drawn from another, $F_B$? We show that 
if we make no assumptions regarding the relation of $F_S$ and $F_B$, then the above result ceases to hold, 
and in fact, fails in the most profound way possible. 
As we show, not only is adding one buyer not enough for the expected welfare of BTR in the augmented market to beat the expected optimal welfare in the original market, 
but in fact for \emph{any number $k$} there exist two distributions $F_S$ and $F_B$ such that adding $k$ buyers is not enough for BTR to beat even the optimal \emph{feasible} mechanism.\footnote{Moreover, the result fails even for some pair of \emph{regular} distributions, a condition that is used in the proof of the BK result.}
Moreover, obtaining any constant fraction of the gains-from-trade of the optimal feasible mechanism is also not possible.
We then strengthen this result and show that this negative result holds not only for BTR, but in fact also for \emph{any other} anonymous\footnote{In the absence of any prior, it is natural to treat all agents the same, and thus anonymity is a natural assumption, and one may even claim that it is in a sense a prerequisite for simplicity.} robust deterministic mechanism, and for any combination of adding any number of sellers in addition to any number of buyers:
\begin{proposition}[See \cref{obs-all-mech-reg-not-suff}]\label{intro-no-sd-negative}
	Let $M$ be any anonymous robust deterministic mechanism. 
	For every $\varepsilon>0$ and for every $\mS$, $\mB$, $\ell$, and $k$,
	there exist two distributions $F_S$ and $F_B$ such that
	\[M(\mS+\ell,\mB+k)<\varepsilon\cdot\sbopt(\mS,\mB)=\varepsilon\cdot \opt(\mS,\mB).\]
\end{proposition}
This strong negative result is essentially rooted at the same reason that the BK result fails when the value of the buyer is rarely above the value of the seller.
To overcome this problem, on top of regularity,
BK also assumes that the buyers are ``serious bidders,'' that is, that the valuation of each of them for the item is at least the value of the seller (which they normalize to $0$). We generalize this assumption to the case in which the sellers' values are private and sampled from a distribution, by assuming that $F_B$ \emph{first-order stochastically dominates} (\emph{FSD}, or simply \emph{stochastically dominates}) $F_S$. Indeed, in case the seller cost is fixed, this assumption precisely becomes the ``serious bidders'' assumption of BK. 
Thus, in the remainder of our analysis we move to study settings in which the buyers' values are sampled from a distribution $F_B$ that (first-order) stochastically dominates $F_S$, the distribution from which the sellers' values are sampled.

\paragraph{One seller stochastically dominated by one buyer:} We start with the case of bilateral trade, i.e., one seller with one item, and one buyer, where the distribution $F_S$ from which the seller's value is drawn is
stochastically dominated by the distribution $F_B$ from which the buyer's value is drawn. We ask whether our result from the i.i.d.\ case can be extended to the stochastic dominance case. That is, with stochastic dominance, is it enough to add one buyer so that BTR in the augmented market will achieve higher expected welfare than the optimal-yet-infeasible VCG mechanism in the original market. We give a negative result to this hope:  we prove that adding one buyer is not enough for BTR to beat the optimum. At this point the reader may wonder whether some other anonymous robust deterministic mechanism could do better, and we once again give a negative answer to this hope, showing that this impossibility holds also for any other such mechanism:\footnote{We emphasize once again that all our results, whether positive or negative, for adding buyers can be recast as analogous results for adding sellers instead. For instance,  recasting \cref{intro-11-negative} this way shows that for any anonymous and robust deterministic mechanism $M$ there exist $F_S$ and $F_B$ such that $F_B$ stochastically dominates $F_S$ and for which $M(2,1)<\sbopt(1,1)=\opt(1,1)$. See \cref{app:sellers-buyers} for the full details.}\textsuperscript{,}\footnote{Readers who are familiar with the \emph{Median Mechanism} of \citeN{McAfee08} may wonder whether a natural variant of that mechanism can obtain in the one-seller-two-buyer market expected welfare that beats the optimum in the original market.  This mechanism, while not robust, requires only limited knowledge of the distributions, in the form of the medians of the distributions, and provides good approximation guarantees under a condition implied by stochastic dominance. Without going into further detail about this mechanism at this point, we note that in \cref{app:mcafee} we give a negative answer even to this hope.}

\begin{theorem}[See \cref{thm:general-mech-lb-12}]\label{intro-11-negative}
Let $M$ be any anonymous robust deterministic mechanism.
	There exist two distributions, $F_S$ and $F_B$,
	such that
	$F_B$ stochastically dominates $F_S$ and 
	for which 
	\[M(1,2)<\sbopt(1,1)=\opt(1,1).\]
\end{theorem}

Having shown that this impossibility result is a general one, we continue focusing on BTR, and now ask whether adding not one but a small fixed number of buyers can guarantee that BTR in the augmented market beats (in the same sense of expected welfare) the optimal-yet-infeasible VCG mechanism in the original market.  To better develop an intuition for this question, we take a quick detour through a related item-pricing question, which is also interesting in its own right. (It also motivates looking at bilateral trade settings separately, and, as we will see, at the BTR mechanism).

\paragraph{An aside: pricing at a fresh sample:} Arguably the most celebrated of the corollaries of the BK result, within the economics and computation community, is that in a scenario with one seller and one buyer (whose value is drawn from a distribution satisfying BK's assumptions), if the seller can obtain one fresh (independent) sample drawn from the buyer's value distribution, then pricing the item at this sample (that is, giving the buyer a take-it-or-leave-it offer to buy at this sample) yields as expected revenue at least half of the expected revenue of the optimal mechanism that uses full knowledge of the buyer's value distribution \cite{DhangwatnotaiRY10}, which by the analysis of \citeN{Myerson81} prices the item at a price carefully optimized for that distribution.

Inspired by this conclusion of BK's result, we give the first results that we know of for sample-pricing in two-sided markets. We first show via a reduction similar to that of \citeN{DhangwatnotaiRY10} that \cref{intro-iid} implies an analogous result in a bilateral trade setting: if the buyer's value and the seller's value are \emph{drawn i.i.d.},\ and if a fresh sample can be obtained from the same distribution, then giving both of them a take-it-or-leave-it offer (which they must both take for trade to happen) to trade at the price of the sample attains in expectation at least half of the gains-from-trade of the optimal-yet-infeasible VCG mechanism:\footnote{We emphasize that attaining a $C$ approximation in terms of the gains-from-trade (the difference between post-trade and pre-trade social welfare) is strictly more challenging than attaining the same $C$ approximation in terms of the social welfare.} 

\begin{theorem}[See \cref{cor:iid-sample-half}]\label{intro-sample-iid}
For a single buyer and a single seller, whose values are drawn i.i.d.,\ pricing at a single fresh sample drawn from the same distribution obtains expected gains-from-trade at least $\frac{1}{2}\opt(1,1)$. Furthermore, there exist distributions for which it obtains expected gains-from-trade exactly $\frac{1}{2}\opt(1,1)$ (and no more than that).
\end{theorem}

What about the non-i.i.d.\ case where we only assume that the buyer's value distribution stochastically dominates the seller's? Can we still get a similar result despite \cref{intro-iid} breaking in this setting? We give an affirmative answer, showing that a qualitatively similar result (though with a provably smaller constant, and with a different proof) still holds: if a fresh sample can be obtained from the \emph{buyer's} value distribution, then giving the seller and the buyer a take-it-or-leave-it offer to trade at the price of the sample attains in expectation at least \emph{one quarter} of the gains-from-trade of the optimal-yet-infeasible VCG mechanism:\footnote{It is interesting to compare this mechanism with the \emph{Median Mechanism} of \citeN{McAfee08}, which under a condition implied by (i.e., weaker than) first-order stochastic dominance obtains in expectation \emph{one half} (in contrast to the one quarter from \cref{intro-pricing-sd}) of the optimal gains-from-trade. Pricing at a sample, as we show, gives a qualitatively similar guarantee of a small constant factor to the optimal gains-from-trade (in \cref{app:mcafee} we show that the one-half guarantee that \citeN{McAfee08} proves indeed cannot be improved upon to prove a better guaranteed constant for that mechanism) under a stronger assumption (stochastic dominance) but while requiring far less prior knowledge: a single sample rather than the precise median. This can be viewed as demonstrating a tradeoff of sorts between the strength of the assumption on the properties of the unknown distribution and the amount of concrete numeric knowledge regarding this distribution that is needed.}

\begin{theorem}[See \cref{thm:FSD-11-sample,prop:FSD-11-sample-nohalf}]\label{intro-pricing-sd}
For a single buyer whose value is drawn from $F_B$ and a single seller whose value is drawn from $F_S$ where 
$F_B$ stochastically dominates $F_S$,
pricing at a single fresh sample drawn from $F_B$ obtains expected gains-from-trade at least $\frac{1}{4}\opt(1,1)$. Furthermore, the fraction $\frac{1}{4}$ in this statement cannot be improved to any fraction higher than $\frac{7}{16}$ (and in particular, cannot be improved to~$\frac{1}{2}$, the fraction that is obtained for i.i.d.\ agents), even if $\opt(1,1)$ is replaced with $\sbopt(1,1)$.
\end{theorem}

\paragraph{Pricing using fresh samples vs.\ adding buyers:} Pricing at a sample, like adding buyers, seems like an intuitive thing to do. Indeed, returning to our question above, it may intuitively seem that adding, say, $10$, or to be on the safe side, say, $100$, more buyers should ``surely suffice'' for BTR in the augmented market to beat the optimal welfare in the original one-seller-one-buyer market. To emphasize the elusiveness of this intuition, we note that this ``intuitively surely working'' claim whereby adding, say, $k=100$ more buyers suffices to beat the optimal mechanism in the original market, implies a far less intuitive claim: that taking $k$ fresh independent samples and giving a take-it-or-leave-it offer to trade at the \emph{highest} of the $k$ samples attains a $\frac{1}{1+k}$ approximation to the gains-from-trade of the optimal-yet-infeasible mechanism in the original one-seller-one-buyer market. (Taking $k=1$, incidentally, derives \cref{intro-sample-iid} from \cref{intro-iid}.)

As the former claim implies the latter,
if the former ``surely holds'' (as it may intuitively seem), then the latter should ``surely hold'' as well, however intuition for the latter is more elusive: on the one hand, a $\frac{1}{k+1}$ approximation may not seem very challenging, but on the other hand, pricing at the maximum of $k$ samples seems like a completely absurd thing to do! In fact, at a first glance at the statement of the latter claim it is not even clear that increasing $k$ makes things any easier (while for the former claim this is completely obvious). Indeed, intuition may be misleading, and the former claim is not as straightforward as it may seem at first glance. (Indeed, recall that without stochastic dominance we have shown that there is no fixed finite number of buyers that if added, allows any anonymous robust deterministic mechanism to beat the optimal in the original market.) Nonetheless, we do manage to show that with one seller and one buyer (under stochastic dominance), if we can recruit not one but \emph{four more} buyers into the market, then BTR in the augmented market would beat the optimal-yet-infeasible mechanism in the original market:
\begin{theorem}[See \cref{prop:sd-1-1-ub}]
	For any seller distribution $F_S$ and any buyer distribution $F_B$ such that 
	$F_B$ stochastically dominates $F_S$,
	it holds that \[\btr(1,1+4)\geq\opt(1,1).\]
\end{theorem} 

\paragraph{One seller stochastically dominated by many buyers:} We move to consider the more general (but not yet most general) case in which a single seller with value drawn from $F_S$ is facing $\mB$ buyers, all of whose values are drawn from $F_B$, which stochastically dominates $F_S$. Will adding a constant number of extra buyers be enough for BTR (or any other anonymous robust deterministic mechanism) to beat the optimum? We show that this is not the case: if the number of buyers we need to add so that BTR, or any other anonymous robust deterministic mechanism, would beat even the optimal \emph{feasible} mechanism is finite, then it must grow with $\mB$, the number of buyers in the original market:
\begin{theorem}[See \cref{thm:sd-1-r-lb}]\label{intro-1r-negative} Let $M$ be any anonymous robust deterministic mechanism.
	For any $k$ there exists a number $N$ such that for any number of buyers $\mB>N$ there exist two distributions, $F_B$ and $F_S$, such that    
	$F_B$ stochastically dominates $F_S$ and
	for which 
	\[M(1,\mB+k)<\sbopt(1,\mB)=\opt(1,\mB).\]
\end{theorem}
For a given fixed $\mB$, though, is adding a finite number of buyers sufficient? (Recall that we have shown in \cref{intro-no-sd-negative} that if the stochastic dominance assumption is dropped, then no finite number of buyers suffices.) Here we give an affirmative answer --- our second main result, once again justifying our focus on BTR: adding an order of $\sqrt{\mB}$ buyers allows BTR in the augmented market to beat the \emph{optimal-yet-infeasible} mechanism in the original market:
\begin{theorem}[See \cref{prop:sd-1-r-ub}]\label{intro-1r-positive}	
	For any number of buyers $\mB$, any seller distribution $F_S$, and any buyer distribution $F_B$ such that $F_B$ stochastically dominates $F_S$, 
	it holds that \[\btr(1,\mB+4\sqrt{\mB})\geq\opt(1,\mB).\]
\end{theorem} 

\cref{intro-1r-negative,intro-1r-positive} together imply that the more buyers there are in the original market, while it is still possible to add buyers so that BTR in the augmented mechanism beats the optimal-yet-infeasible mechanism in the original market, the more buyers we need to accomplish this. This is quite striking given that we show that in fact, the more buyers there are in the original market, the better the approximation that BTR guarantees to the gains-from-trade of the optimal-yet-infeasible mechanism without adding any buyers to begin with, and moreover, the guaranteed GFT of BTR approaches the optimal GFT in the limit when $\mB$ grows large:
\begin{theorem}[See \cref{btr-1r-converge}]\label{intro-1r-approx}
\[\left(\inf_{\substack{F_S,F_B~\text{s.t.}\\ F_B~\text{FSD}~F_S}}\frac{\btr(1,\mB)}{\opt(1,\mB)}\right)\xrightarrow[~~\mB\rightarrow\infty~~]{}1.\]
\end{theorem}

It is worth emphasizing that the approximation guarantee given in \cref{intro-1r-approx} departs from known approximation guarantees for mechanisms similar to BTR \cite{McAfee92} as it holds with a \emph{single} seller (under the stochastic domination assumption), a setting in which the size of trade does not grow large.

\paragraph{Many sellers stochastically dominated by many buyers:} Finally, we return to the setting of a double-auction
market with arbitrary numbers of buyers and sellers that we first analyzed in \cref{intro-iid}, however with the buyers' values drawn from a distribution that stochastically dominates (rather than equals as in \cref{intro-iid}) that from which the sellers' values are drawn. For this setting, we ask once again whether there is any finite number that is a function of the number of buyers and the number of sellers (but not of the distributions) such that adding so many buyers to the market guarantees that BTR in the augmented market beats the optimal-yet-infeasible mechanism in the original market. We give an affirmative answer, showing that such a finite number exists:

\begin{theorem}[See \cref{thm:sd-m-r-ub}]\label{intro-mr}
	For any number of sellers $\mS$, any number of buyers $\mB$, any seller distribution $F_S$, and any buyer distribution $F_B$ such that 	$F_B$ stochastically dominates $F_S$,
	it holds that
	\[\btr(\mS,\mB+\mS\cdot (\mB+4\sqrt{\mB}))\geq\opt(\mS,\mB).\]
\end{theorem}
The proof of \cref{intro-mr} is based on our result for a single seller (\cref{intro-1r-positive}), combined with some new observations, which may also be of independent interest.

The upper bound on the required number of added buyers that we present in \cref{intro-mr} is rather high, and, as noted above this theorem, should first and foremost be viewed as a qualitative result --- that some \emph{finite} number of additional buyers suffices uniformly over all distributions (given stochastic domination).  It may also be viewed as the first step in quantifying the number of added buyers that is necessary and 
sufficient to beat the optimum for any pair of distributions (under stochastic dominance). We thus leave the problem of lowering this upper bound and getting tight quantitative results as our main open problem:

\begin{open-problem}
Tighten our understanding of the necessary and sufficient number of added buyers,
\[
k(\mS,\mB)=\min\bigl\{k\in\mathbb{N}~\big|~\forall F_S,F_B~\text{s.t.}~F_B~\text{FSD}~F_S:\btr(\mS,\mB+k)\ge\opt(\mS,\mB)\bigr\}.
\]
(\cref{intro-mr} shows that the set on the right-hand side is nonempty and therefore that $k(\mS,\mB)$ is a well defined finite number.)
\end{open-problem}

\noindent
The main BK-style results of this paper are summarized in \cref{table:bkresults} in the beginning of this section and the sampling results are summarized in \cref{table:sampleresults} there. Proofs are relegated to the appendix.

\subsection{Additional Related Work} \label{sec:related-wrok}

\paragraph{Bulow-Klemperer (BK) style results:} There are many examples of BK-style results for revenue in one-sided markets.  Results in single-dimensional settings include \citeN{HartlineR09,DRS09}.  \citet{HartlineR09} show that for non-identical regular distributions, recruiting one extra bidder per distribution guarantees a 3-approximation for downward-closed feasibility constraints.
\citet{sivan2013vickrey} extend this to non-identical irregular distributions, deconstructing each irregular distribution into a convex combination of regular distributions, and adding an additional bidder from each underlying regular distribution.  Running VCG with these additional bidders achieves a 2-approximation.  Recently, \citet{fu2019vickrey} show that in the non-identical regular setting, with a small amount of information about each distribution (e.g., a sample, or the median), one can then pick a single distribution to duplicate and run VCG with to achieve a $10$-approximation to the optimal revenue in the original setting.
Some recent work present results for a single item in a dynamic setting \cite{LiuP16}.  A recent trend has been extending the BK result to multi-dimensional one-sided auction settings, where a seller auctions multiple heterogeneous items to multiple buyers.  The first result of this form was by \citeN{RoughgardenTY12} for unit-demand bidders over independent, regular items.  The work by \citeN{EdenFFTW16b} found the first such result that beats the optimal revenue
for the additive setting, while  \citeN{FFR18} improve the bounds on the number of added buyers drastically, but they only recover a $(1-\vareps)$-fraction of the optimal revenue. Very recent work of \citeN{BW19} match these bounds in the small market regime and drastically improve the bounds of \citeN{EdenFFTW16b} in the large market regime, even though the \cite{BW19} results are for precise coverage of the optimal revenue.
While all the above papers presents BK-style results for revenue maximization in \emph{one-sided} markets, our work presents BK-style results for welfare (and gains from trade) in \emph{two-sided} markets, which to the best of our knowledge was not suggested or studied in any prior paper.

\paragraph{Approximation Mechanisms:} A prominent approach in the Algorithmic Game Theory community to circumvent the \citeN{MyersonS83} impossibility result is to aim for welfare or GFT \emph{approximation} by truthful mechanisms.  
Several recent papers \cite{BlumrosenD16, BlumrosenM16, BrustleCFM17,BCGZ18,Colini-Baldeschi16, Colini-Baldeschi17, Colini-BaldeschiGKSRT17} have indeed taken this approach of designing mechanisms that obtain welfare or GFT approximations in both bilateral trade settings and in generalizations of it. Yet, in sharp contrast to our paper, the mechanisms in all of these papers are tailored to the prior distributions (as a robust mechanism must be a posted-price mechanism and thus cannot given any approximation guarantee), while we design robust mechanisms that have no access to the distributions yet are able to beat the optimal welfare. We are able to obtain this at  the cost of assuming that we can recruit some additional buyers from the buyer distribution (in the spirit of BK).

\paragraph{Trade Reduction (TR) Mechanisms:} The seminal work of \citeN{McAfee92} introduced a double-auction mechanism that is IR, truthful, and budget-balanced, and has high efficiency in markets with large trade size. That mechanism and its Trade Reduction variant (in which one trade is always reduced) were generalized to various domains, including supply chains \cite{BabaioffN04,BabaioffW05},  spatially-distributed markets \cite{BabaioffNP09} and matching markets \cite{BCGZ18}. \citeN{DuttingRT14} presented a modular approach to the design of double 
 auctions and used the Trade Reduction mechanism within it. \citeN{Bredin05} have generalized the Trade Reduction mechanism to dynamic settings.  

\paragraph{Non-Truthful Double Auctions:} 
Classic work in economics \cite{SatterthwaiteW89, RustichiniSW94, SatterthwaiteW2002} has considered the equilibrium outcome of non-truthful mechanisms.  
These papers studied the rate of convergence to efficiency in any Bayesian equilibrium of non-truthful double-auction mechanisms\footnote{The mechanisms considered  set a clearing price in the interval of prices that would result in efficient trade with respect to the reports.} when distributions are fixed and the market grows large. 
In contrast, we are interested in truthful mechanisms and in obtaining results even for small markets.

\paragraph{Prior-Independent Mechanisms:} 
The assumption that the mechanism designer knows the distributions from which values are sampled is a strong one and commonly not realistic.  
The Wilson doctrine \cite{Wilson85} advocates mechanisms that use as little information as possible about the priors. 
The prior-independent framework assumes that a prior exists, yet is not known to the designer, and measures the performance of a mechanism with respect to the (unknown) prior.  
Following early results on prior-independent mechanisms such as \citeN{segal2003optimal}, the systematic study of such mechanisms was initiated by \citeN{HartlineR09}. For results on prior-independent mechanisms see, for example, Chapter 5 of \citeN{HartlineBook}, as well as \citeN{devanur2011prior}. All of the above results are for monopolistic revenue maximization.
A ``hybrid'' approach still very much in line with the Wilson doctrine is to restrict the dependence of the auction mechanism on the full details of the prior, by having it depend only on certain statistical measures of the valuation distribution, such as its mean, its variance, or its median. The most relevant to our paper is the work of \citeN{McAfee08} (see discussions in \cref{sec-11}). For revenue-maximization results with this approach, see \citeN{AzarDMW13,AzarM13}. The mechanism we use (BTR) is completely prior-independent and our results are squarely within the prior-independent framework. 

\paragraph{Prior-Independent Double Auctions:} 
\citet{deshmukh2002truthful} study \emph{prior-free} profit maximization by an  intermediary in a double auction, where no underlying prior is even assumed to exist.  
They reduce this problem to prior-free profit maximization by a monopolist seller, and also extend the techniques known for that setting directly to the double auction setting.  
Following this, \citet{baliga2003market} study prior-independent profit-maximization by an intermediary between a number of i.i.d.\ buyers and i.i.d.\ sellers.  Asymptotically as the market size grows to infinity, a sampling approach gives an estimate of the distributions and can be used to maximize profit in the limit, while still maintaining incentive compatibility.

\paragraph{Single-Sample-Based Mechanisms:} 
Relaxing the assumption that no information at all is available about the prior, but still requiring less information than precise knowledge of statistical measures such as its mean, variance, or median, recent papers on revenue maximization in one-sided markets have considered using one sample (or more) from the buyers distribution by the mechanisms. \citeN{DhangwatnotaiRY10} initiated this literature, showing that in a single-item environment, a single sample from a regular distribution (that is used as the price) is enough to get a 2-approximation to the optimal revenue.  They also established the connection between BK results and single-sample results. \citeN{HuangMR15} proved that the approximation of \citeN{DhangwatnotaiRY10} is optimal for deterministic mechanisms, and improved it for MHR distributions. \citeN{FuILS15} show that the approximation of $2$ for regular distributions can be improved by using randomization. \citeN{goldner2016prior} build on \citeN{DhangwatnotaiRY10} to give constant-approximation single-sample results for additive multi-item environments. \citeN{azar2014prophet} study prophet inequalities under a single-sample setting. \citeN{BGMM18} show that the approximation of $2$ can be improved if two samples, rather than one, are used. Papers at the interface of algorithmic mechanism design and learning theory that use polynomially many samples for revenue approximation in single- and multi-dimensional environments include \citeN{cr14,mr15,mr16,dhp16,rs16,ht16,bsv16,bsv17,s17,gn17,cai2017learning,abgmmy17,gw18}.
As far as we know, all prior such results are for monopolist revenue maximization, while our paper is the first to show that samples can be useful in approximating welfare (and GFT) in two-sided markets.
\section{Preliminaries}

\subsection{Model}

\paragraph{Setting:}
We consider double-auction market settings with buyers and sellers of identical indivisible items. 
Let $S$ denote the set of sellers, and $B$ the set of buyers. We denote the number of sellers by $\mS$ and the number of buyers by $\mB$.
Each seller $i\in S$ can sell one item, has a value $s_i\geq 0 $ for keeping the item, and has a value of $0$ if selling the item. Each buyer $i\in B$ is interested in one item, has value $b_i \geq 0$ for obtaining an item, and has a value of $0$ if not obtaining any item. 
The value $s_i$ of seller $i$ is sampled independently from a distribution $F_S$, and the value $b_i$ of buyer  $i$ is sampled independently from a distribution $F_B$. 
In Section~\ref{sec:iid} we will be interested in the special case where $F_S=F_B$ (i.e., the valuations of the sellers and the buyers are sampled from the same distribution), the \emph{i.i.d.\ case}.

We denote the vector of realized buyers values by $\vecb$, and the vector of realized sellers values by $\vecs$. 
Throughout this paper we break ties in favor of buyers (so if there is a tie between a buyer value and a seller value, the buyer will be considered as having the higher value) and then lexicographically by ID. 
Let $b^{(j)}$ be the value of the $j$\textsuperscript{th} \emph{highest}-order statistic out of $\vecb$, and we use $s^{(j)}$ to denote the value of the $j$\textsuperscript{th} \emph{lowest}-order statistic out of $\vecs$.
We use $x^{(j)}$ to denote the value of the $j$\textsuperscript{th} highest-order statistic out of the union of $\vecb$ and $\vecs$. 
We sometimes abuse notation and use the same notation to also refer to the agent and not only to her value.  When it is not clear what market the order statistics come from, we denote the number of sellers and buyers in the market by writing the order statistic as $\X{j}{\mS}{\mB}$.

\paragraph{Allocations, Welfare, and Gains from Trade:} An \emph{allocation} in a market with $\mS$ sellers and $\mB$ buyers is a specification of $\mS$ agents who hold/receive items.
We say that a buyer who holds an item \emph{trades} in the allocation, and similarly, a seller who does not hold an item \emph{trades} in the allocation. The \emph{trade size} in an allocation is the number of buyers who trade in it (equivalently, the number of sellers who trade in it).
The \emph{(social) welfare} of an allocation is the sum of the values of the $\mS$ agents who hold/receive items. Thus, the \emph{pre-trade welfare} --- the welfare of the \emph{pre-trade allocation} in which each seller holds her item --- is $\sum_{i\in S} s_i$. The \emph{optimal (post-trade) welfare} is therefore $\sum_{j=1}^{\mS} x^{(j)}$. The gains from trade (GFT) of an allocation is the difference between the welfare of that allocation and the pre-trade welfare.
The \emph{optimal gains from trade (optimal GFT)}, denoted by $\opt(\vecb,\vecs)$, therefore equals
$\opt(\vecb,\vecs)= \sum_{j=1}^{\mS} x^{(j)}- \sum_{i\in S} s_i$. 
The \emph{optimal trade size} $q=q(\vecs, \vecb)$
is defined as the trade size in the maximal-trade-size allocation that gives the optimal welfare, and therefore equals the maximal number of buyers such that each of the $q$ highest-value buyers has a value that is at least as high as each of the $q$ lowest value sellers. Equivalently, the optimal trade size is the number of buyers among $x^{(1)},\ldots,x^{(\mS)}$.

\paragraph{Mechanisms:}
We assume that each agent
has a \emph{quasi-linear utility function} --- her utility equals her value for the item that she holds (or $0$ if she holds no item) minus the \emph{payment} that she pays (which is negative if she is in fact paid) --- and that each agent aims to maximize her own utility. 
In this paper we consider only dominant-strategy incentive compatible mechanisms, so by the revelation principle we can, without loss of generality, restrict ourselves to 
direct-revelation mechanisms. 
We focus on \emph{deterministic} mechanisms.\footnote{One could also consider \emph{randomized} mechanisms. 
	All the mechanisms we discuss will be deterministic and for our positive results we will achieve our goal of competing with any mechanism, even those that are not restricted to be deterministic. Our lower bounds are proven only for deterministic mechanisms and we leave open the question of the power of randomized mechanisms.}  
Such a mechanism defines an allocation and a payment for each agent, as functions of the (reported) agents values and of the distributions $F_S$ and $F_B$.
For any mechanism $M$, we use the notation $M(\mS, \mB)$ to denote the expect GFT of $M$ when all agents are truthful, keeping the distributions $F_S$ and $F_B$, over which the expectation is taken, implicit in the notation as they will be clear from the context.

A mechanism is \emph{truthful} (or \emph{incentive compatible (IC)}) if it is a dominant strategy for every agent to report her valuation truthfully.\footnote{As all mechanism we consider are truthful, we always assume that the reported values are the true values, and consider the performance of truthful mechanisms under truthful reporting. } 
A mechanism is \emph{individually rational (IR)} if every agent, when truthful, ends up with nonnegative utility at every outcome of the mechanism.
A mechanism is \emph{(weakly) budget-balanced (BB)} if the sum of all payments from the agents to the mechanism is nonnegative at every outcome of the mechanism.
A mechanism is \emph{feasible} if it is individually rational, truthful, and budget-balanced.
A mechanism is \emph{optimal}
(or \emph{efficient}, or \emph{welfare maximizing}) if it maximizes the welfare (and thus the gains from trade) for every realization of the agents values.
A mechanism is \emph{prior-independent} if its allocation and payment depend only on the reported values and not on $F_S$ or $F_B$.
A mechanism is \emph{robust} if it is feasible and prior-independent.
A truthful mechanism is \emph{anonymous} if applying a permutation to the values of all agents, such that buyer values remain buyer values and seller values remain seller values, causes the same permutation to be applied to the agents' critical values for winning.
In particular, the only property of anonymous mechanisms that we will use is that if no agent trades in a given anonymous mechanism for some value profile,
then if we permute the buyers values among the buyers,
and permute the sellers values among the sellers, 
still no agent trades in this mechanism for the resulting value profile.

\subsection{Specific Mechanisms}

We now turn to define a few specific mechanisms that will be of special interest in our analysis. In each of these mechanisms, agents who do not trade pay 0. 

\paragraph{The optimal-yet-infeasible VCG mechanism:} The Vickrey-Clarke-Groves (VCG) mechanism
is a mechanism that first computes the allocation that maximizes the GFT for the realized values (recall that this GFT is denoted $\opt(\vecs,\vecb)$),  and then charges trading agents their critical value for trading, with non-trading agents paying 0. 
This mechanism is deterministic, IR, truthful, optimal, prior-independent, and anonymous, yet it has a budget deficit\footnote{As an illustration, consider a seller with value $2$ and a buyer with value $3$. To maximize welfare, the critical value for the buyer to trade is $2$, so she should pay $2$, while the critical value for the seller to trade is $3$, so she should receive $3$. Hence, the VCG mechanism runs a deficit of $1$!} and is therefore not feasible and not robust. By slight abuse of notation, we also denote this optimal-yet-infeasible mechanism by $\opt$. (So $\opt(\mS,\mB)$ denotes the expected GFT of the optimal-yet-infeasible VCG mechanism in a market with $\mS$ sellers whose values are drawn from $F_S$ and $\mB$ buyers whose values are drawn from $F_B$.)

\paragraph{The optimal feasible mechanism:} As noted in the introduction, the seminal impossibility result of
\citeN{MyersonS83} implies, in particular,  that for bilateral-trade, that is, for the setting where a single seller ($\mS=1$) wishes to sell a single item to a single buyer ($\mB=1$), there is no mechanism that is IR, truthful, BB, and optimal.\footnote{The actual result is in fact stronger and holds when significantly relaxing the discussed properties as long as IR holds at least interim. For example, it holds for Bayesian IC and not only for dominant strategies IC, and for ex-ante BB and not only ex-post BB. }
The mechanism that maximizes the expected GFT subject to these constraints (of being IR, truthful, BB) is called the \emph{optimal feasible} mechanism (it is also common to refer to this mechanism as the \emph{second-best} mechanism). As noted in the introduction, such a mechanism is known to be very complex even for bilateral trade, and the specification of the mechanism requires detailed information about the exact value distributions (so it is not prior-independent, and therefore not robust).
We denote this feasible-yet-not-robust mechanism by $\sbopt$. (So $\sbopt(\mS,\mB)$ denotes the expected GFT of the optimal feasible mechanism in a market with $\mS$ sellers whose values are drawn from $F_S$ and $\mB$ buyers whose values are drawn from $F_B$.)

\paragraph{The BTR mechanism:}
We next describe the \emph{Buyer Trade Reduction (BTR)} mechanism, which is a variant of the celebrated
mechanism of \citeN{McAfee92}.
Given a realization of the  values $(\vecs,\vecb)$, this mechanism computes $q = q(\vecs, \vecb)$, the size of the optimal trade, and checks whether $b^{(q+1)}\geq s^{(q)}$. If so, then the $q$  sellers with lowest values trade with the $q$ buyers with highest values at the price $b^{(q+1)}$; that is, for every $1\le i\le q$, seller $s^{(i)}$ trades with buyer $b^{(i)}$ for a price of $b^{(q+1)}$. (Recall that by definition, in this case $b^{(i)}\ge b^{(q+1)}\ge s^{(q)}\ge s^{(i)}$, so the mechanism is IR and no deficit arises.) Otherwise, i.e., if $b^{(q+1)}<s^{(q)}$, then the $q\!-\!1$ sellers with lowest values trade with the $q\!-\!1$ buyers with the highest values, where each of these $q\!-\!1$ sellers receives a payment of $s^{(q)}$ and each of these $q\!-\!1$ buyers pays $b^{(q)}$. (Recall that by definition of $q$, we have that $b^{(q)}\ge s^{(q)}$, so no deficit arises.)
While this mechanisms has high efficiency when the trade size in the efficient trade is large, it does not give any guarantee regarding the achieved efficiency when this trade size is very small.\footnote{This is the case when $q=1$, implying that the GFT of BTR is 0 while the optimal GFT might be positive. It happens, in particular, in the bilateral trade setting, in which it is always the case that $q\leq 1$, so there is never trade under BTR, and the GFT obtained is always 0. The GFT of BTR is also 0 in settings with more agents, when $q=1$ and the second-highest-value buyer has lower value than the lowest-value seller.} We denote this anonymous and robust\footnote{See \cref{sec:btr-ic} for a proof of the truthfulness of BTR.} deterministic mechanism by $\btr$. (So $\btr(\mS,\mB)$ denotes the expected GFT of the BTR mechanism in a market with $\mS$ sellers whose values are drawn from $F_S$ and $\mB$ buyers whose values are drawn from $F_B$.)

When there is only one seller, the BTR mechanism can be even more naturally described as follows: 
the mechanism runs a second price auction with no reserve price between the buyers, and if the price determined by the auction is acceptable to (i.e., at least the value of) the seller, then trade happens between the seller and the winning buyer of the auction (i.e., the highest-value buyer), at the price determined by the auction (i.e., at the second-highest buyer value as the price). Note that this indeed coincides with BTR for the case of a single seller: in BTR there is trade between the highest-value buyer and the seller, priced at the second-highest buyer value, if and only if both of these agents accept this price, i.e., if and only if the second-highest buyer value is at least the seller value (by definition, the highest buyer value is at least the second-highest buyer value, so the highest-value buyer always accepts this price).

\subsection{Additional Notation}\label{sec:prelim-quantiles}

\paragraph{Quantiles and Stochastic Dominance:}
For a distribution $F$ and \emph{quantile} $q \in [0,1]$, we define the \emph{value} corresponding to the quantile $q$ as the unique value
\[v_F(q) = \inf\bigl\{v \mid \Pr_{w\sim F}[w \leq v] \geq q\bigr\}.\]
Note that for any distribution, $v_F$ is monotone nondecreasing, that is, higher quantiles correspond to (weakly) higher values.
It is well known that for any distribution $F$, drawing a quantile uniformly at random from $(0,1)$ and then taking the value that corresponds to that quantile results in a value distributed according to $F$.
A distribution $F_B$ is said to \emph{first-order stochastically dominate} (\emph{FSD}, or simply \emph{stochastically dominate}) a distribution $F_S$, denoted by $F_B~\text{FSD}~F_S$, if $v_{F_B}(q) \geq v_{F_S}(q)$ for every $q\in[0,1]$.
\section{Independently and Identically Distributed Buyers and Sellers}\label{sec:iid}

\subsection{A General Positive BK-Style Result}

We start with our first main positive result, showing that if the values of all agents (buyers and sellers) are sampled i.i.d.\ from the same distribution, then adding a single buyer is enough: in expectation, the welfare (resp.\ GFT) of the anonymous and robust Buyer Trade Reduction mechanism with one more buyer is at least as high as the expected optimum welfare (resp.\ GFT), achievable by the optimal-yet-infeasible (not budget balanced) VCG mechanism.

\begin{theorem} \label{thm:iid}
Consider any setting with $\mS$ sellers and $\mB$ buyers whose values are sampled identically and independently from some distribution $F$ (that is, $F=F_S=F_B$).
Then the Buyer Trade Reduction (BTR) mechanism with one additional buyer has expected GFT that is at least as large as the expected optimal GFT before adding the extra buyer. \\
That is, for any setting with i.i.d.\ buyers and sellers it holds that  \[\btr(\mS,\mB+1)\geq \opt(\mS,\mB).\]
\end{theorem}

We stress that \cref{thm:iid} places no assumptions whatsoever on the distribution $F$, and in particular, does \emph{not} assume the regularity condition. This is in contrast to the celebrated result of \citeN{BK} that shows that for buyers with values sampled i.i.d.\ from a \emph{regular} distribution and a seller with no cost,  
the \emph{revenue} of a second-price auction (with no reserve) with one additional buyer independently sampled from the same distribution is at least as high the revenue of the optimal \citeN{Myerson81} mechanism before adding that buyer. That result indeed does not hold for some distributions that are not regular, while our result does. While the results are analogous (for different settings and different objectives, of course), the proofs of the two results are very different and do not seem to be related.

Many of our proofs regarding BTR (including the proof of \cref{thm:iid}) use a characterization of the GFT of BTR compared to the optimum captured by the next lemma. It essentially says (and when there are no ties, precisely says) that the only realizations for which BTR does not achieve the optimal welfare (equivalently, does not achieve the optimal GFT), are those in which a seller has the $(\mS\!+\!1)^{\text{st}}$ realized highest value. 

\begin{lemma}
\label{lem:reduce}
In a double-auction market with $\mS$ sellers and $\mB$ buyers, the Buyer Trade Reduction (BTR) mechanism achieves optimal GFT if and only if some buyer's value equals $x^{(\mS+1)}$, the $(\mS\!+\!1)^{\text{st}}$ realized highest value.\footnote{We emphasize that in the case of ties between values, the condition is that one of the agents whose realized values equal the $(\mS\!+\!1)^{\text{st}}$ realized highest value is a buyer.} Furthermore, if BTR does not achieve optimal GFT, then the efficient trade size $q$ is positive, and the GFT of BTR are lower than optimal by $b^{(q)}-x^{(\mS+1)}$.
\end{lemma}

We briefly outline the coupling argument at the heart of the proof of \cref{thm:iid} (for the complete proof see \cref{app:iid}). We prove the claim by comparing each side of the inequality in the \lcnamecref{thm:iid} statement to the optimal gains from trade in the augmented market with $\mS$ sellers and $\mB+1$ buyers, and showing that 
\begin{equation}\label{thm:iid:inequality}
\opt(\mS,\mB+1) - \opt(\mS,\mB) \geq \opt(\mS,\mB +1) - \btr(\mS,\mB+1),
\end{equation}
which clearly implies the claim.
We relate the loss from $\opt(m_S, m_B+1)$ in both sides of \cref{thm:iid:inequality} by a coupling argument that hinges on the fact that as agent values are drawn i.i.d., any permutation of the values of the different agents is equally likely. Concretely, we couple the two random markets as follows: first draw $m_S+ m_B+1$ values, then---only for construction of the smaller market---remove one value uniformly at random, and then (in each market) uniformly at random assign $m_S$ of the remaining values to be seller values, and the rest of the remaining values to be buyer values. Having defined this coupling of the randomness of the two sides of \cref{thm:iid:inequality}, we prove this \lcnamecref{thm:iid:inequality} for any realization of the draw of $m_S+m_B+1$ values, in expectation only over the uniformly random assignment of these values to the various agents.\footnote{So we henceforth slightly abuse notation by considering \cref{thm:iid:inequality} to be defined for some fixed realization of these values, with the expectations taken only over the random assignment of values to agents in each market.} Indeed, for any given realization of the $\mS+\mB+1$ values, we show that the two sides of \cref{thm:iid:inequality} are non-zero with the same probability, and that conditioned upon being non-zero, the difference on the l.h.s.\ is larger.

For the l.h.s.\ of \cref{thm:iid:inequality}, first note that since the expected pre-trade welfare is the same in both markets (the expected sum of a random selection of $m_S$ of the values), the difference in expected optimal GFTs is the difference between the expected post-trade welfares. To compare these between the two markets, compare the above random process without the removal step and with that step. The post-trade welfare in each market is the sum of the top $m_S$ agent (seller or buyer) values in that market. These two sums differ precisely when the removed value is selected to be one of the top $m_S$ values, which occurs with probability $m_S/(m_S+m_B+1)$, and then the difference is the removed value minus the original $(\mS\!+\!1)$\textsuperscript{st} value.

For the r.h.s.\ of \cref{thm:iid:inequality}, consider the above random process without the removal step. By \cref{lem:reduce}, BTR loses compared to $\opt$ precisely when the $(\mS\!+\!1)\textsuperscript{st}$ highest-valued agent is a seller. This too occurs with probability $m_S/(m_S+m_B+1)$, and then the difference is, by the same \lcnamecref{lem:reduce}, the minimum of the top $q$ buyer values, again minus the original $(\mS\!+\!1)$\textsuperscript{th} value. Canceling out the $(\mS\!+\!1)$\textsuperscript{st} value on both sides, we have that conditioned upon the l.h.s.\ being non-zero, in that side we have a \emph{uniformly} random value selected from the top $m_S$ agent values, while conditioned upon the r.h.s.\ being non-zero, in that side we have the \emph{minimum} of the top $q\ge1$ buyer values, which are \emph{uniformly} distributed among the top $m_S$ agent values, so the latter is of course smaller, which is what we set out to prove.  For the full details, see the proof in the appendix.

\subsection{Beyond I.I.D.\ Sellers and Buyers: A Negative Result in the Absence of Additional Assumptions}

We next consider weakening the assumption made in \cref{thm:iid} that the valuations of both buyers and sellers are sampled from the same distribution (that is, that $F_S=F_B$).
We next show that if we make no assumptions whatsoever about the relation of $F_B$ and $F_S$, then the result of \cref{thm:iid} ceases to hold, 
and in fact fails in the most profound way possible. This happen even if the distributions are regular, 
a condition that was used in the result of \citeN{BK}.
We show that not only adding one buyer is not enough for the expected welfare of BTR in the augmented market to beat the expected optimal welfare,
but in fact for \emph{any number $k$} there exist two distributions $F_S$ and $F_B$ such that adding $k$ buyers is not enough for BTR to beat even the optimal \emph{feasible} mechanism. Moreover, obtaining any constant fraction of the GFT of the optimal feasible mechanism is also not possible, and this holds even if we also allow to add any number of sellers.

\begin{proposition} \label{obs-reg-not-suff}For every $\varepsilon>0$ and every positive integers $\mS$, $\mB$, $\ell$, and $k$,
	there exist two regular distributions $F_S$ and $F_B$ such that
	\[\btr(\mS+\ell,\mB+k)<\varepsilon\cdot \sbopt(\mS,\mB)=\varepsilon\cdot  \opt(\mS,\mB).\]
\end{proposition}
One may wonder if this strong negative result is only due to the fact that we are using BTR and not some other mechanism. 
We give a negative answer by showing that a similar result (without the regularity condition, though) holds for any robust deterministic mechanism that is furthermore anonymous.\footnote{In the absence of any prior, it is natural to treat all agents the same, and thus anonymity is a natural assumption, and one may even claim that it is in a sense a prerequisite for simplicity.}
\begin{proposition} \label{obs-all-mech-reg-not-suff}
	For any prior-independent mechanism $M$ that is deterministic, IR, truthful, weakly budget-balanced, and anonymous, the following holds. 
	For every $\varepsilon>0$ and for every positive integers $\mS$, $\mB$, $\ell$, and $k$,
	there exist two distributions $F_S$ and $F_B$ such that
	\[M(\mS+\ell,\mB+k)<\varepsilon\cdot \sbopt(\mS,\mB)= \varepsilon\cdot \opt(\mS,\mB).\]
\end{proposition}
So we see that the result of \cref{thm:iid} falls very far from extending without any assumptions on the relation of the two distributions $F_S$ and $F_B$, 
not only for BTR but for a large class of robust mechanisms.
This strong negative result is essentially rooted at the same reason that the result of \citeN{BK} fails when the value of the buyer is rarely above the value of the seller (even for a regular distribution). 
Indeed, it can be shown that their result would also fail on the distributions used in the proof of \cref{obs-reg-not-suff}, as with these distributions the value of any buyer is rarely above the value of the seller. 
To overcome this problem, as noted in the introduction, \citeauthor{BK} also assume that the bidders are ``serious bidders,'' that is, that the valuation of each of them is nonnegative, i.e., at least the value of the seller that they normalized to $0$. We generalize this assumption to the case in which the sellers' values are private and sampled from a distribution, by assuming that $F_B$ (first-order) stochastically dominates $F_S$. Indeed, in case the seller cost is fixed, this assumption precisely becomes the ``serious bidders'' assumption of \citeN{BK}.
Thus, in the remaining sections we move to study settings in which the buyers' values are sampled from a distribution $F_B$ that (first-order) stochastically dominates $F_S$, the distribution from which the sellers' values are sampled.
\section{One Seller and One Buyer (Bilateral Trade)}\label{sec-11}

In this section we consider the bilateral trade setting, i.e., the setting of one seller with one item, and one buyer.

\subsection{Negative Results}\label{sec:11-negative}

Surprisingly for us, even in the simple case of bilateral trade, the result presented in \cref{thm:iid} for the setting of all agents values are sampled i.i.d., 
does not extend even to the setting where 
the distribution of the buyer's values $F_B$  stochastically dominates the distribution of the seller's value $F_S$.
As we show, not only does the expected welfare of BTR with one more buyer not beat the expected optimal welfare, but it also does not even beat the expected welfare of the optimal feasible  mechanism (which may depend on the distributions).
As we show, there exist two distributions, a buyer distribution $F_B$ that  stochastically dominates some seller distribution $F_S$, such that 
for bilateral trade, adding one buyer sampled from $F_B$ is not enough for BTR to beat the welfare of the optimal feasible mechanism (and thus the optimal welfare). 
 
\begin{proposition} \label{prop:sd-1-1-lb}	
There exist two distributions, $F_S$ and $F_B$, such that $F_B$ stochastically dominates $F_S$ and
for which \[\btr(1,2)<\sbopt(1,1)= \opt(1,1).\]
\end{proposition}

Given the failure of BTR to beat the optimum with one more buyer, one may wonder whether some other anonymous robust deterministic mechanism for one seller and two buyers could do better than BTR, and succeeds in beating the optimal welfare for one seller and one buyer. We give a negative answer to this hope in \cref{thm:general-mech-lb-12} below. The proof of that theorem utilizes \cref{prop:sd-1-1-lb} above, and relates the welfare guarantee of any anonymous robust deterministic, mechanism with that of BTR via the following key lemma, which shows that with or without stochastic dominance, if any anonymous robust deterministic mechanism has welfare higher than BTR on any profile, then it must fail miserably for some distributions.

\begin{lemma} \label{obs-all-mech-reg-not-suff-12}
	For any prior-independent mechanism $M$ for the setting of one seller and two buyers,  
	that is deterministic, IR, truthful, weakly budget-balanced, and anonymous, the following holds. 
	If there is any value profile for which the welfare of $M$ is higher than the welfare of $BTR$, then 
	for every $\varepsilon>0$ there exist two distributions, $F_S$ and $F_B$, such that  $F_B$ stochastically dominates $F_S$ 
	and for which
\[M(1,2)<\varepsilon\cdot \sbopt(1,1)= \varepsilon\cdot \opt(1,1).\]
\end{lemma}

Since by \cref{prop:sd-1-1-lb}, BTR with one more buyer does not beat the optimum in the original market, and by \cref{obs-all-mech-reg-not-suff-12} any other robust and anonymous deterministic mechanism $M$ can never guarantee more that BTR across all distributions for the setting with one seller and two buyers, the following theorem immediately follows.
\begin{theorem}\label{thm:general-mech-lb-12}
	For any prior-independent mechanism $M$ for the setting of one seller and two buyers,
	that is deterministic, IR, truthful, weakly budget-balanced, and anonymous, the following holds. 
	There exist two distributions, $F_S$ and $F_B$, such that $F_B$ stochastically dominates $F_S$ 
	and for which \[M(1,2)<\sbopt(1,1)= \opt(1,1).\]
\end{theorem} 


\citeN{McAfee08} has presented a mechanism\footnote{Note that this mechanism is unrelated to \citeauthor{McAfee92}'s mechanism for double-auction markets \cite{McAfee92} that we discussed earlier.} for GFT approximation for bilateral trade settings in which the median of the buyer's distribution is at least as high as the median of the seller's distribution. This clearly is the case when $F_B$ stochastically dominates $F_S$. \citeauthor{McAfee08}'s \emph{Median Mechanism} is not prior-independent (and hence not robust), however it requires only limited numeric knowledge regarding $F_B$ and $F_S$, in the form of their medians, or even only the median of one of them. This mechanism posts some price between these two medians, and if both agents accept the price, trade occurs. \citeauthor{McAfee08} proved that this mechanism has expected GFT that is at least half of the optimum.

One might be tempted to try using (a variant of) \citeauthor{McAfee08}'s median mechanism when there are two buyers and a seller, 
aiming  to get a result similar to the one we presented in \cref{thm:iid} for BTR, even only for the special case of bilateral trade.
What would a reasonable interpretation of the median mechanism with two buyers in the i.i.d.\ case be? 
We suggest that the most natural variant posts the median as the price, and if the seller and \emph{at least one} of the two buyers agree to that price, 
then the seller trades with the \emph{higher-value} buyer. This mechanism is indeed truthful and strongly budget balanced as one may desire, and while not prior-independent, requires only limited numeric knowledge regarding $F_B$ and $F_S$.
Yet, we show that even with this extra knowledge, the ``median mechanism with one additional buyer'' does not beat the optimal GFT even for the special case 
that all values are sampled i.i.d. ($F_S=F_B$).
Contrast this with \cref{thm:iid}, where in the same setting, we have shown that BTR with one additional buyer does beat the optimal GFT in a the corresponding one-buyer-one-seller market.
As we show in \cref{app:mcafee}, there exists a distribution such that when the seller's value and two buyers' values are sampled i.i.d.\ from that distribution,
the ``median mechanism with one additional buyer'' has GFT that is less than the optimum GFT with only one buyer and one seller. 
Alternatively stated, even for bilateral trade and i.i.d.\ agents, the median mechanism with one additional buyer does \emph{not} beat the optimum, while BTR does. 

Having shown that the impossibility of beating the optimum with only one added buyer (\cref{thm:general-mech-lb-12}) is general to all anonymous robust deterministic mechanisms and is not a unique drawback of BTR, we continue focusing on BTR, and now move toward asking whether adding not one but a small fixed number of buyers can guarantee that BTR in the augmented market beats the optimum in the original market.  To better develop an intuition for this question, we take a quick detour through a related item-pricing question, which is also interesting in its own right (and also motivates looking at bilateral trade settings separately, and, as we will see, at the BTR mechanism).

\subsection{Pricing using Fresh Samples}

Arguably the most celebrated corollaries of BK's result, within the economics and computation community, is that in a scenario with one seller and one buyer (whose value is drawn from a distribution satisfying BK's assumptions), if the seller can obtain one fresh (independent) sample drawn from the buyer's value distribution, then pricing the item at this sample (that is, giving the buyer a take-it-or-leave-it offer to buy at this sample) yields expected revenue at least half of the expected revenue of the optimal mechanism that uses full knowledge of the buyer's value distribution \cite{DhangwatnotaiRY10}, which by the analysis of \citeN{Myerson81} prices the item at a price carefully optimized for that distribution.

Inspired by this conclusion of BK's result, we first show, via a reduction similar to that of \citeN{DhangwatnotaiRY10}, a causal relation between the ability of BTR with two buyers and one seller to beat the optimum in a one-seller-one-buyer, and the ability of pricing at a sample from $F_B$ in bilateral trade to obtain expected GFT at least half of the optimum:\footnote{We emphasize that attaining a $C$ approximation in terms of the gains-from-trade (the difference between post-trade and pre-trade social welfare) is strictly more challenging than attaining the same $C$ approximation in terms of the social welfare.} 

\begin{lemma}\label{obs:one-sample}
	For any seller distribution $F_S$ and any buyer distribution $F_B$, if the following holds:
	\begin{itemize}
	\item
	$\btr(1, 2) \geq \opt(1,1)$,
	\end{itemize}
	then the following also holds:
	\begin{itemize}
	\item
	In a setting with one seller and one buyer, pricing at a single sample drawn independently from $F_B$ obtains expected GFT at least $\frac{1}{2} \opt(1,1)$.
	\end{itemize}
	Furthermore, if $F_B$ is atomless, then  the converse implication is also true: if the latter statement holds, then so does the former.\footnote{\label{footnote-sampling}We note that for some distributions with atoms the converse implication is actually false. Indeed, for seller with value $1$ and buyers with value uniform in $\{0,2\}$ it holds that $\opt(1,1)=1/2>1/4=\btr(1,2)$ violating the first condition, yet pricing at the a sample (maximum of $1$ samples) drawn independently from $F_B$ obtains expected GFT of $1/4$ which equals $\frac{1}{2} \opt(1,1)$, satisfying the second condition. 	}
 \end{lemma}
\cref{obs:one-sample} is a special case of a more general result that we prove later (\cref{obs:k-samples}) and we thus omit the proof for this special case.

From \cref{obs:one-sample,thm:iid} we immediately derive the first result, to the best of our knowledge, on sample pricing in two-sided markets: that if the buyer's value and the seller's value are drawn i.i.d.,\ and if a fresh sample can be obtained from the same distribution, then giving both of them a take-it-or-leave-it offer (which they must both take for trade to happen) to trade at the price of the sample attains in expectation at least half of the GFT of the optimal-yet-infeasible mechanism, and we furthermore show that the constant one half cannot be improved upon:

\begin{theorem}\label{cor:iid-sample-half}
	For a single buyer and a single seller with values drawn i.i.d.\ from $F=F_S=F_B$, pricing at a single independent sample drawn from $F$ obtains expected GFT at least $\frac{1}{2} \opt(1,1)$. Furthermore, if $F$ is atomless, then the obtained expected GFT is exactly $\frac{1}{2} \opt(1,1)$ (and no more than that).
\end{theorem}

What about the non-i.i.d.\ case where we only assume that the buyer's value distribution stochastically dominates the seller's?
Can we still get a qualitatively similar result to \cref{cor:iid-sample-half} for the stochastic dominance setting, though? We give an affirmative answer, showing that while pricing at a fresh sample from $F_B$ does not always attain in expectation at least half of the GFT of even the optimal infeasible mechanism, it does attain in expectation at least one quarter of the GFT of even the optimal-yet-infeasible mechanism:

\begin{theorem} \label{thm:FSD-11-sample}
	For a single buyer with value drawn from $F_B$ and a single seller with value drawn from $F_S$ where $F_B~\text{FSD}~F_S$, pricing at a single independent sample drawn from $F_B$ obtains expected GFT at least $\frac{1}{4} \opt(1,1)$.\end{theorem}

\begin{proposition}\label{prop:FSD-11-sample-nohalf}
	For every $\alpha>\frac{7}{16}$ there exist two distributions, a seller distribution $F_S$ and a buyer distribution $F_B$, such that $F_B~\text{FSD}~F_S$ and for which for a single buyer with value drawn from $F_B$ and a single seller with value drawn from $F_S$, pricing at a single independent sample drawn from $F_B$ obtains expected GFT of less than $\alpha\cdot\sbopt(1,1)=\alpha\cdot\opt(1,1)$.
\end{proposition}

It is interesting to compare this ``pricing at a fresh sample'' mechanism with the \emph{Median Mechanism} of \citeN{McAfee08}, which, as noted above, under a condition that is implied by (i.e., is weaker than) stochastic dominance, obtains in expectation \emph{one half} of the optimal GFT, in contrast to the one quarter from \cref{thm:FSD-11-sample}.  (In \cref{app:mcafee} we show that the one-half guarantee that \citeN{McAfee08} proves indeed cannot be improved upon to prove a better guaranteed constant for that mechanism.)  Pricing at a sample, as we show, gives a qualitatively similar guarantee of a small constant factor to the optimal GFT  under a stronger assumption (stochastic dominance) but while requiring far less prior knowledge: a single sample rather than the precise median. This can be viewed as demonstrating a tradeoff of sorts between the strength of the assumption on the properties of the unknown distribution and the amount of concrete numeric knowledge regarding this distribution that is needed.

Pricing at a sample, like adding buyers, seems like an intuitive thing to do. Indeed, returning to our question posed at the end of \cref{sec:11-negative}, it may intuitively seem that adding, say, $10$, or to be on the safe side, say, $100$, more buyers should ``surely suffice'' for BTR in the augmented market to beat the optimal welfare in the original one-seller-one-buyer market. To emphasize the elusiveness of this intuition, we note that this ``intuitively surely working'' claim whereby adding, say, $k=100$ more buyers suffices to beat the optimal mechanism in the original market, implies a far less intuitive claim: that taking $k$ fresh independent samples and giving a take-it-or-leave-it offer to trade at the \emph{highest} of the $k$ samples attains a $\frac{1}{1+k}$ approximation to the optimal GFT in the original one-seller-one-buyer market:

\begin{lemma} \label{obs:k-samples}
	For any seller distribution $F_S$ and any buyer distribution $F_B$, if the following holds:
	\begin{itemize}
	\item
	$\btr(1, 1+k) \geq \opt(1,1)$,
	\end{itemize}
	then the following also holds:
	\begin{itemize}
	\item
	In a setting with one seller and one buyer, pricing at the maximum of $k$ samples drawn independently from $F_B$ obtains expected GFT at least $\frac{1}{1+k} \opt(1,1)$.
	\end{itemize}
	Furthermore, if $F_B$ is atomless, then the converse implication is also true: if the latter statement holds, then so does the former.\footnote{Yet, for some distributions with atoms, the converse implication is actually false, see \cref{footnote-sampling}. 	}
 \end{lemma}

As the first statement in \cref{obs:k-samples} implies the second, if the first statement ``surely holds'' (as it may intuitively seem), then the second statement should ``surely hold'' as well, however intuition for the latter is more elusive: on the one hand a $\frac{1}{k+1}$ approximation may not seem very challenging, but on the other hand pricing at the maximum of $k$ samples seems like a completely absurd thing to do! In fact, at a first glance at the second statement, it is not even clear that increasing $k$ makes things any easier (while for the first statement this is completely obvious). Indeed, intuition may be misleading, and the first statement is not as straightforward as may seem at first glance. (Indeed, recall that without stochastic dominance we have shown in \cref{obs-all-mech-reg-not-suff} that there is no fixed finite number of buyers that if added, allows any anonymous robust deterministic mechanism to beat the optimal in the original market.) Nonetheless, we proceed as planned to study the question of whether, and for which $k$, the first statement holds.

\subsection{A Positive Result}

Despite \cref{prop:sd-1-1-lb} (and more generally, \cref{obs-all-mech-reg-not-suff-12}) showing that adding one buyer is not enough for BTR in the augmented market to beat the optimum in the original market, and despite the intuitive uncertainty raised by the equivalence that we have proven in \cref{obs:k-samples}, we next show that for any buyer distribution $F_B$ that stochastically dominates the seller distribution $F_S$,
adding not one buyer but \emph{four more buyers} to the market is enough for BTR in the augmented market to beat the optimum in the original market:

\begin{sloppypar}
\begin{proposition} \label{prop:sd-1-1-ub}
	For any seller distribution $F_S$ and any buyer distribution $F_B$ such that $F_B~\text{FSD}~F_S$ it holds that \[\btr(1,1+4)\geq\opt(1,1).\]
\end{proposition}
\end{sloppypar}

\cref{prop:sd-1-1-ub} is a special case of a more general result that we prove later for the case of many buyers (\cref{prop:sd-1-r-ub}) and we thus omit the proof for this special case.

The qualitative message of \cref{prop:sd-1-1-ub,obs-all-mech-reg-not-suff-12} is that with stochastic dominance, adding some constant number of buyers is enough for guaranteeing that BTR beats the optimum in the original market, but that this constant is greater than $1$. We leave open the question of nailing down the exact constant --- be it $2$, $3$, or $4$.
\section{One Seller and Many Buyers}\label{sec:1r}

In \cref{prop:sd-1-1-ub}, we saw that for the case of bilateral trade where the buyer distribution stochastically dominates the seller distribution, adding a constant number of buyers is enough for BTR in the augmented market to beat the optimum in the original market. Can this result be generalized to any number of buyers and sellers as in the i.i.d.\ case (\cref{thm:iid})? We start by showing that this is not possible already for the case of a single seller and $\mB$ buyers, showing that in this case the number of buyers we need to add must grow with the initial number of buyers $\mB$, and in particular, it is not constant. Specifically, we show that for some distributions we must add at least $\lfloor\log_2 \mB\rfloor$ more buyers for BTR in the augmented market to beat the optimum in the original market. This shows that, in particular, the number of buyers that we have to add cannot be independent of the number of initial buyers, and moreover, it cannot simply be a function only of the number of sellers. Like with our other lower bounds, we once again show that is not merely a unique shortcoming of the BTR mechanism, but in fact a fundamental limitation of any anonymous robust deterministic mechanism, by showing that any mechanism that beats BTR on the distributions on which BTR underperforms must miserably fail on some other distributions.

\begin{theorem} \label{thm:sd-1-r-lb}
For any  prior-independent mechanism $M$ that is deterministic, IR, truthful, weakly budget-balanced, and anonymous the following holds. 
For any positive number $\mB$ there exist two distributions, $F_B$ and $F_S$, such that 
$F_B$ stochastically dominates  $F_S$ and for which 
\[M(1,\mB+\lfloor\log_2 \mB\rfloor)<\sbopt(1,\mB)= \opt(1,\mB).\]
\end{theorem}

\cref{thm:sd-1-r-lb} implies that if a finite number of buyers can be added so that BTR in the augmented market beats the optimum in the original market, then this number of buyers grows with $\mB$ (implying \cref{intro-1r-negative}). But for a given fixed $\mB$, is adding a finite number of buyers actually sufficient? (Recall that we have shown in \cref{obs-all-mech-reg-not-suff} that if the stochastic dominance assumption is dropped, then no finite number of buyers suffices.)  We now complement \cref{thm:sd-1-r-lb} with our second main result, this time for stochastic dominance, showing that for the case of a single seller, there is some function of the initial number of buyers $\mB$, but not of the distributions, that bounds the number of additional buyers that is sufficient to add for BTR in the augmented market to beat the optimum in the original market.  Specifically, we show that adding $4\sqrt{\mB}$  more buyers suffices.

\begin{proposition} \label{prop:sd-1-r-ub}	
	For any number of buyers $\mB$, any seller distribution $F_S$ and any buyer distribution $F_B$ that stochastically dominates  $F_S$, it holds that \[\btr(1,\mB+4\sqrt{\mB})\geq\opt(1,\mB).\]
\end{proposition}

Now that the agents are not all sampled i.i.d. (the seller comes from a different distribution) we cannot use the same idea used in the proof of \cref{thm:iid}, that the assignment of \emph{values} to agents is uniformly random. Instead, we argue that we can first draw quantiles and then assign the \emph{quantiles} to the agents uniformly at random; we then use stochastic dominance, from which follows that for the same quantile the buyer's value always exceeds the seller's, to prove the claim. For details see the proof in the appendix.

A gap still remains between the lower bound of \cref{thm:sd-1-r-lb} and the upper bound of \cref{prop:sd-1-r-ub}, and we leave the problem of closing this gap as an interesting open problem.

These two bounds together imply that the more buyers there are in the original market, while it is still possible to add buyers so that BTR in the augmented mechanism beats the optimum in the original market, the more buyers we need to accomplish this. This is quite striking given that we now show that in fact, the more buyers there are in the original market, the better the approximation that BTR guarantees to the expected GFT\footnote{Thus also to the optimal welfare as GFT approximation implies welfare approximation.} of the optimum to begin with, and moreover, as the number of buyers tends to infinity this approximation tends to $1$ (vanishing loss):

\begin{theorem}\label{btr-1r-converge}
	For any number of buyers $\mB$, any seller distribution $F_S$, and any buyer distribution $F_B$ such that $F_B$ stochastically dominates $F_S$
	it holds that 
	\[\btr(1,\mB)\ge\frac{\mB-1}{\mB+1} \cdot \opt(1,\mB).\]
\end{theorem}

\noindent \cref{intro-1r-approx} directly follows from \cref{btr-1r-converge} (for the proof of the theorem see \cref{app:btr-converge}.)

It is worth emphasizing that the approximation guarantee given in \cref{btr-1r-converge} departs from known approximation guarantees in the CS literature for mechanisms similar to BTR (\cite{McAfee92}'s mechanism and the Trade Reduction mechanism \cite{BabaioffN04}). For these mechanisms, it was observed that for any realization $(\vecs,\vecb)$, these mechanisms get at least a $\frac{q-1}{q}$ fraction of the optimum, where $q=q(\vecs,\vecb)$ is the efficient trade size. 
With a single seller, $q\leq 1$, so $q$ does not grow large even in a ``large market'' as even though the number of buyers $\mB$ grows large there is only one seller, 
and the above fraction is always zero, providing no approximation guarantee at all. 
In contrast, the approximation we present is a function of the size of the input ($\mB$) and holds when there is only one seller (the case that $q\leq 1$),
as long as $F_B$ stochastically dominates $F_S$. Moreover, it converges to the optimum as the number of buyers grows large, 
while any approximation that is a function of $q$ does not converge to the optimum when the number of sellers does not grow (as $q\leq \min \{\mS,\mB\}$).
\section{Many Sellers and Many Buyers}\label{sec:mr}

In this section we consider the  general setting of a double-auction market with arbitrary numbers of buyers and sellers that we first analyzed in \cref{sec:iid}, however with the buyers' values drawn from a distribution $F_B$ that stochastically dominates (rather than equals as in \cref{sec:iid}) the distribution $F_S$ from which the sellers' values are drawn. We ask whether \cref{prop:sd-1-r-ub} can be generalized to this setting: whether there is any finite number that is a function of the number of buyers and the number of sellers (but not of the distributions) such that adding so many buyers to the market guarantees that BTR in the augmented market beats the optimum in the original market. We give an affirmative result, showing that such a finite number indeed exists:

\begin{theorem} \label{thm:sd-m-r-ub}
	For any number of sellers $\mS$, any number of buyers $\mB$, any seller distribution $F_S$ and any buyer distribution $F_B$ that stochastically dominates  $F_S$, it holds that
	\[\btr(\mS,\mS\cdot (\mB+4\sqrt{\mB}))\geq\opt(\mS,\mB).\]
\end{theorem}

As noted in the introduction, \cref{thm:sd-m-r-ub} should indeed first and foremost be viewed as a qualitative result --- that some \emph{finite} number of additional buyers suffices uniformly over all distributions (given stochastic domination) --- and as the first step in quantifying the number of added buyers that is necessary and sufficient to beat the optimum for any pair of distributions (under first-order stochastic dominance). As noted in the introduction, we thus leave the problem of lowering this bound and getting tight quantitative results as our main open problem.

We  prove \cref{thm:sd-m-r-ub} by carefully reducing to the single-seller case of \cref{prop:sd-1-r-ub}. Before we can spell out this reduction, we will need to develop a property of $\opt$ and a property of $\btr$. 
We believe both of these properties, and especially the latter one (the property of $\btr$), to be also of independent interest, and thus present it below (for the other see the appendix). 
For BTR we show that if there are multiple markets running BTR, unifying them into a single BTR market will never reduce the expected GFT.   

\begin{lemma}\label{lem:btr-m-sum-unified}
	For every $F_B$ and $F_S$ it holds that for every positive integers $m_1,m_2,\ldots,m_t$ and $r_1,r_2,\ldots,r_t$:
	\[\sum_{i=1}^t \btr(m_i,r_i) \leq  \btr\left(\sum_{i=1}^t m_i,\sum_{i=1}^t r_i\right).\]
\end{lemma}

\noindent
Using \cref{lem:btr-m-sum-unified} 
we can now prove \cref{thm:sd-m-r-ub} by reducing to the single-seller case of \cref{prop:sd-1-r-ub}:

\begin{proof}[Proof of \cref{thm:sd-m-r-ub}]
		By \cref{prop:sd-1-r-ub} it holds that $\btr(1,\mB+4\sqrt{\mB})\geq\opt(1,\mB)$. 
		We combine this with \cref{lem:btr-m-sum-unified} (for
		$m_i=1, r_i=\mB+4\sqrt{\mB}$ for every $i$) and with $\opt(\mS,\mB)\leq \mS\cdot \opt(1,\mB)$ (property of $\opt$ that we prove in the appendix)
		to prove the theorem:
	\[\opt(\mS,\mB)\leq \mS\cdot \opt(1,\mB)\leq \mS\cdot \btr(1,\mB+4\sqrt{\mB}) \leq  \btr(\mS,\mS(\mB+4\sqrt{\mB})).\tag*{\qedhere}\]
\end{proof}

\section{Concluding Remarks}
In this paper, we applied Bulow-Klemperer-style analysis beyond monopolist revenue maximization, and have presented BK-style results for welfare in double-auction markets. We have suggested using the robust Buyer Trade Reduction mechanism and running it with additional buyers to beat the optimum welfare in the original market.  
We have shown that when the values of all agents are sampled i.i.d., BTR with just one more buyer beats the optimum in the original market, irrespective of the value distribution and of the original numbers of buyers and sellers. A possible direction for future research is extending these results beyond the double-auction setting. We then moved to study the more challenging setting in which the values of the buyers are sampled from one distribution, while the values of the sellers are sampled from another. We have shown that with no assumptions about the two distributions, for any anonymous robust deterministic mechanism, no matter how many buyers we add, the mechanism will fail to beat the optimum for some distributions. We thus assumed that the buyers' distribution stochastically dominates the sellers' distribution.

Focusing on the case of one buyer and one seller (with stochastic dominance), we have shown that adding one more buyer is not enough for any anonymous robust deterministic mechanism to beat the optimum, yet BTR with a constant number of additional buyers beats the optimum for any distribution. We have also related BTR to pricing using samples. We then studied the case of many buyers and one seller, and have shown that for any anonymous robust deterministic mechanism, the number of buyers that we have to add must grow with the number of initial buyers.  However, when running BTR, there is some finite number of buyers, which moreover grows significantly slower than the number of original buyers, that is sufficient to add, so that BTR in the augmented market beats the optimum welfare in the original market. This is somewhat surprising, as we also show that as the number of buyers grows, BTR with no additional buyers converges to the optimum. Finally, we have shown that for general double-auction markets, with any number of sellers and any number of buyers, a BK-style result for BTR with stochastic dominance is indeed possible, showing once again that it is sufficient to add some finite number of additional buyers which does not depend on the value distributions. We leave open the problem of pinning-down the exact number of buyers that need to be added.

\appendix
\section{Adding Sellers}\label{app:sellers-buyers}

The following proposition allows us to to recast each of our results for adding buyers to an analogous result for adding sellers, and vice versa for any result for adding sellers.

\begin{proposition}\label{sellers-buyers}
There is a one-to-one and onto correspondence $i$ between signed vectors of realized values $(\vecs, \vecb)$ and signed vectors of realized values $i(\vecs,\vecb)=(i(\vecb),i(\vecs))$ (so the original seller values determine the new buyer values and vice versa) such that all of the following hold:
\begin{itemize}
\item
For every vector of realized values $(\vecs, \vecb)$, all of the following hold:
\begin{itemize}
\item
$\opt(i(\vecs,\vecb))=\opt(\vecs, \vecb)$,
\item
$\sbopt(i(\vecs,\vecb))=\sbopt(\vecs,\vecb)$,
\item
$\str(i(\vecs,\vecb))=\btr(\vecs,\vecb)$.
\end{itemize}
More generally, for every mechanism $M$ for $m_S$ sellers and $m_B$ buyers there exists a mechanism $M'$ for $m_B$ sellers and $m_S$ buyers such that $M'(i(\vecs,\vecb))=M(\vecs,\vecb)$ for every vector of realized values $(\vecs, \vecb)$. Furthermore, whichever of the following properties that is satisfied by~$M$ is also respectively satisfied by $M'$: determinism, IR, truthfulness, weak budget balance, anonymity.
\item
Applying $i$ pointwise, letting $F'_S=i(F_B)$ (recall that this is the distribution of the new \emph{seller} values) and $F'_B=i(F_S)$ (recall this is the distribution of the new \emph{buyer} values), we have that $F_B~\text{FSD}~F_S$ if and only if $F'_B~\text{FSD}~F'_S$.
\end{itemize}
\end{proposition}

We emphasize that the correspondence in \cref{sellers-buyers} is between \emph{signed} vectors of realized values and \emph{signed} vectors of realized values. In particular, positive buyer values are mapped to negative seller values and vice versa. For bounded distributions this can be ``remedied'' by shifting all resulting values (of sellers and of buyers) by a constant (which does not affect the GFT). For unbounded distributions we note that in fact none of our proofs rely on the sign of any value, and so they carry through this correspondence. The proof of \cref{sellers-buyers} is immediate: simply define the new \emph{seller} values as ${i(\vecb)=-\vecb}$ and the new \emph{buyer} values as $i(\vecs)=-\vecs$, and all of the claims in \cref{sellers-buyers} follow immediately from definitions. (In particular, since we both flip the signs \emph{and} flip the roles, stochastic dominance of the buyers' distribution over the sellers' is maintained.) As added intuition, one may think of the items in the new market as ``not owning an item of the old market''; this is indeed what buyers ``bring to the market'' and what sellers ``buy from buyers'' in trade (and indeed, if the original items are goods, i.e., have positive value, then the new items are bads, i.e., have negative value, and vice versa).

\section{McAfee's Median mechanism} \label{app:mcafee}
\citeN{McAfee08} has presented a mechanism for GFT approximation for bilateral trade when the median of the buyer's distribution is at least as high as the median of the seller's distribution (which is implied by our assumption that the buyer's distribution FSD the seller's.) 
The mechanism posts any price between the two medians, and if both agents accept the price, trade happens.
When the seller and the buyer distributions are identical, it simply posts the median of that distribution.\footnote{We remark that the mechanism $M'$ that corresponds to the median mechanism via the correspondence from \cref{app:sellers-buyers} is the median mechanism itself, and so the impossibility result below also holds for adding a seller rather than a buyer.}
\citeauthor{McAfee08} proved that this mechanism has expected GFT that is at least half of the optimum. We show below 1)~that this guarantee is the best that could be given for this mechanism, and 2)~that the natural extension of this mechanism for one seller and two buyers, which we call ``median mechanism with one additional buyer'' --- where the median is posted as the price, and if the seller and at least one of the two buyers agree to that price, the seller trades with the higher valued buyer --- does not beat the optimum in the original market even in i.i.d.\ settings (while as we have shown, BTR does). See the discussions in \cref{sec-11} for additional context.

\begin{lemma}
	For distribution $F$, denoting the expected GFT of the median mechanism with one seller and one buyer by $\median(1,1)$, and the expected GFT for the median mechanism with one additional buyer  by $\median(1,2)$.  it holds that:
	\[\opt(1,1) \geq \delta  - O(\delta^2),\]
	\[\median(1,1) \leq  \frac{\delta}{2} + O(\delta^2),\]
	\[\median(1,2) \leq \frac{7}{8}\cdot \delta + O(\delta^2).\]	
	Thus, for small enough $\delta$: 
	\[\median(1,2)<\opt(1,1).\]
	Additionally, for every $c>1/2$, for small enough $\delta$:
	\[\median(1,1)<c\cdot  \opt(1,1).\]
\end{lemma}
\begin{proof}
The distribution $F$ that we will use is the following distribution:  
\[v = \begin{cases} 0 & w.p.\ \delta/2   \\ 1-\delta^{2} & w.p.\ 1/2 - \delta/2 - \delta^{10}  \\ 1& w.p.\ 2\delta^{10}  \\ 1+ \delta^{2}& w.p. \ 1/2 - \delta/2- \delta^{10} \\ 2& w.p.\ \delta/2  \end{cases}\]
Note that the median of this distribution is at $1$. 

	With two agents, any event for which none of them has value in $\{1-\delta^2, 1+\delta^2\}$ has negligible probability ($O(\delta^2)$) and therefore negligible contribution ($O(\delta^2)$) to any of the above expectations. Additionally, the event that none have value in $\{0,2\}$ has negligible gains from trade ($O(\delta^2)$) and therefore negligible contribution ($O(\delta^2)$) to any of the above expectations.

	We first compute the optimal GFT with one seller and one buyer. The first term considers the case that the GFT is  $1+\delta^2$ (obtained either when the buyer has value of 2 and seller has value of $1-\delta^2$, or the buyer has value of $1+\delta^2$ and the seller has value $0$.)
	while the second is the case that the GFT is $1-\delta^2$ (obtained either when the buyer has value of 2 and seller has value of $1+\delta^2$, or the buyer has value of $1-\delta^2$ and the seller has value $0$.). For each we have two symmetric cases:
	\[\opt(1,1) \geq 2 (1+\delta^2) \cdot \frac{\delta}{2}\cdot \frac{1-\delta}{2} + 2(1-\delta^2) \frac{\delta}{2}\cdot \frac{1-\delta}{2} - O(\delta^2)=  \delta- O(\delta^2).\]
	
	For the median mechanism with one seller and one buyer, there is trade only if the buyer has value at least 1, while the seller at most 1. So in this case the only non-negligible  case is the one in which  the GFT is  $1+\delta^2$. There are two such events with equal probability. 
	We get: 
	\[\median(1,1) \leq 2 (1+\delta^2) \cdot \frac{\delta}{2}\cdot \frac{1-\delta}{2} + O(\delta^2)=  \frac{\delta}{2} + O(\delta^2).\]
	
	For the median mechanism with one additional buyer,	there is trade only if the seller has value at most 1, while the highest seller has value at least 1, in any such case trade happens with the highest buyer. Again, the only  non-negligible  case is the one in which  the GFT is  $1+\delta^2$.
	Up to low order terms, with probability $\delta$ we have at least one buyer with value $2$, and we get the contribution when the seller has value $1-\delta^2$, with probability about half. The other  non-negligible  case is the one that at least one buyer has value $1+\delta^2$, and the seller has value 0. This happens with probability about $\frac{3}{4}\cdot\frac{\delta}{2}$. 
	We get 
	\[\median(1,2) \leq (1+\delta^2)\cdot ((1-(1-\frac{\delta}{2})^2)\cdot \frac{1-\delta}{2}   + \frac{\delta}{2}\cdot (1-(1-(\frac{1-\delta}{2}))^2) ) +
	O(\delta^2) = \frac{7}{8}\cdot \delta + O(\delta^2).\tag*{\qedhere}
	\]
\end{proof}

\section{Truthfulness of Buyer Trade Reduction}\label{sec:btr-ic}

\begin{proposition}
The BTR mechanism is dominant-strategy incentive compatible.
\end{proposition}

\begin{proof}
To observe that BTR is dominant-strategy incentive compatible it is enough to show that it is monotone in the bids (that is, a trading agent cannot be removed from the trade by bidding more competitively, i.e., bidding higher if she is a buyer, or lower if she is a seller) and that the payments are by critical values.
Recall that $q$ denotes the efficient trade size. 

We first consider any trading buyer with a bid $b$ and consider two cases, when $b^{(q+1)}\geq s^{(q)}$ holds and when it does not. In either case we denote by $q$ the efficient trade size when the buyer bids $b$ and by $\bar{q}$ the efficient trade size when the buyer bids a different bid $\bar{b}$. In the latter case we denote the $j$\textsuperscript{th} ordered statistic of the bids of the buyers by $\bar{b}^{(j)}$ and of the sellers by $\bar{s}^{(j)}=s^{(j)}$.
If $b^{(q+1)}\geq s^{(q)}$ holds, then the trade size in BTR is $q$, so since the buyer is trading it must be the case that $b=b^{(i)}$ for some $i\in \{1,2,\ldots,q\}$, so $b\geq  b^{(q+1)}$. Consider any bid $\bar{b}> b^{(q+1)}$ by the buyer. It would still be one of the $q$ highest buyer bids among the new bids. It would still hold that $\bar{b}^{(q+1)}=b^{(q+1)}<s^{(q+1)}=\bar{s}^{(q+1)}$ as well as that $\bar{b}^{(q+1)}=b^{(q+1)}\geq s^{(q)}=\bar{s}^{(q)}$, so $\bar{q}=q$ and the new trade size in BTR is $\bar{q}=q$. Therefore, by bidding $\bar{b}> b^{(q+1)}$, the buyer would still trade. On the other hand, any bid $\bar{b}< b^{(q+1)}$ would no longer be one of the $q$ highest buyer bids among the new bids, and the new trade size in BTR would not be greater than $\bar{q}\le q$. As the new trade size is not greater than $q$ and the buyer bidding $\bar{b}$ is not one of the $q$ highest bidding buyers, the buyer would not be trading with any bid $\bar{b}< b^{(q+1)}$. 
We conclude that in the case that $b^{(q+1)}\geq s^{(q)}$ holds, $b^{(q+1)}$ is the critical value for the buyer to trade, and equal to her payment.

We turn to the other case, when  $b^{(q+1)}< s^{(q)}$. In this case the trade size in BTR is $q-1$, so since the buyer is trading it must be the case that $b=b^{(i)}$ for some $i\in \{1,2,\ldots,q-1\}$, so $b\geq b^{(q)}$. We argue that $b^{(q)}$ is the critical value for the buyer to trade. Indeed, bidding any $\bar{b}>b^{(q)}$ would still be one of the highest $q-1$ buyer bids among the new bids. It would still hold that $\bar{b}^{(q+1)}=b^{(q+1)}<s^{(q)}=\bar{s}^{(q)}$ as well as that $\bar{b}^{(q)}=b^{(q)}\ge s^{(q)}=\bar{s}^{(q)}$, so $\bar{q}=q$ and the new trade size in BTR is $\bar{q}-1=q-1$, and we conclude that the buyer would still trade with any bid $\bar{b}>b^{(q)}$. On the other hand, any bid $\bar{b}<b^{(q)}$ would no longer be one of the $q-1$ highest buyer bids among the new bids. In this case, if either $\bar{q}<q$ or this bid is not one of the $q$ highest buyer bids among the new bids, then the buyer would not be trading. Otherwise, $\bar{q}=q$ and this is the $q$ highest buyer bid among the new bids, so since $\bar{b}^{(\bar{q}+1)}=b^{(q+1)}< s^{(q)}=\bar{s}^{(\bar{q})}$, we have that the new trade size in BTR is $\bar{q}-1=q-1$, so the buyer would not be trading in this case as well. 
We conclude that in the case that $b^{(q+1)}< s^{(q)}$ holds, $b^{(q)}$ is the critical value for the buyer to trade, and equal to her payment. This completes the proof that for buyers, payments are by critical values. 

Next, consider any trading seller with a bid $s$, and once again consider the same two cases. In either case, we denote by $q$ the efficient trade size when the seller bids $s$ and by $\bar{q}$ the efficient trade size when the seller bids a different bid $\bar{s}$. In the latter case we denote the $j$\textsuperscript{th} ordered statistic of the bids of the sellers by $\bar{s}^{(j)}$ and of the buyers by $\bar{b}^{(j)}=b^{(j)}$.
If $b^{(q+1)}\geq s^{(q)}$ holds, then the trade size in BTR is $q$, so since the seller is trading it must be the case that $s=s^{(i)}$ for some $i\in \{1,2,\ldots,q\}$, so $s\leq s^{(q)}\le b^{(q+1)}$.
Consider any bid $\bar{s}<b^{(q+1)}$ by the seller. It would still be one of the $q$ lowest seller bids among the new bids since $b^{(q+1)}<s^{(q+1)}$ by definition of $q$. It would still hold that $\bar{b}^{(q+1)}=b^{(q+1)}<s^{(q+1)}=\bar{s}^{(q+1)}$ and that $\bar{s}^{(q)}\le\max\{s^{(q)},\bar{s}\}\le b^{(q+1)}=\bar{b}^{(q+1)}$, so $\bar{q}=q$ and the new trade size in BTR is $\bar{q}=q$. Therefore, by bidding $\bar{s}<b^{(q+1)}$, the seller would still trade. On the other hand, any bid $\bar{s}>b^{(q+1)}$ would no longer be one of the $q-1$ lowest seller bids among the new bids. In this case, if either $\bar{q}<q$ or this bid is not one of the $q$ lowest seller bids among the new bids, then the seller would not be trading. Otherwise, $\bar{q}=q$ and this is the $q$ lowest seller bid among the new bids, so since $\bar{s}^{(\bar{q})}=\bar{s}>b^{(q+1)}=\bar{b}^{(\bar{q}+1)}$, we have that the new trade size in BTR is $\bar{q}-1=q-1$, so the seller would not be trading in this case as well. We conclude that in the case that $b^{(q+1)}\geq s^{(q)}$ holds, $b^{(q+1)}$ is the critical value for the seller to trade, and equal to her payment.

Finally, assume that $b^{(q+1)}< s^{(q)}$. In this case the trade size in BTR is $q-1$, so since the seller is trading it must be the case that $s=s^{(i)}$ for some $i\in \{1,2,\ldots,q-1\}$, so $s\leq s^{(q)}$. We argue that $ s^{(q)}$ is the critical value for the seller to trade. Indeed, bidding any $\bar{s}<s^{(q)}$ would still be one of the lowest $q-1$ seller bids among the new bids. It would still hold that $\bar{b}^{(q+1)}=b^{(q+1)}<s^{(q)}=\bar{s}^{(q)}$ as well as that $\bar{b}^{(q)}=b^{(q)}\ge s^{(q)}=\bar{s}^{(q)}$, so $\bar{q}=q$ and the new trade size in BTR is $\bar{q}-1=q-1$, and we conclude that the seller would still trade with any bid $\bar{s}<s^{(q)}$. On the other hand, any bid $\bar{s}> s^{(q)}$ would no longer be one of the $q-1$ lowest seller bids among the new bids. In this case, if either $\bar{q}<q$ or this bid is not one of the $q$ lowest seller bids among the new bids, then the seller would not be trading. Otherwise, $\bar{q}=q$ and this is the $q$ lowest seller bid among the new bids, so since $\bar{s}^{(\bar{q})}=\bar{s}>s^{(q)}>b^{(q+1)}=\bar{b}^{(\bar{q}+1)}$, we have that the new trade size in BTR is $\bar{q}-1=q-1$, so  the seller would not be trading in this case as well.
We conclude that in the case that $b^{(q+1)}< s^{(q)}$ holds, $s^{(q)}$ is the critical value for the seller to trade, and equal to her payment. This completes the proof that for sellers, payments are by critical values. 
\end{proof}

\section{Lemmas used in Proofs of Lower Bounds}\label{app:lower-bounds}

We next prove a lemma that will enable us to extend all the lower bounds that we prove for BTR to also hold for any other anonymous robust deterministic mechanism.
Recall that for buyer distribution $F_B$ and seller distribution $F_S$, we denote by $M(\mS,\mB)$ the expected GFT of a mechanism $M$ for $\mS$ sellers sampled i.i.d. from $F_S$ and  $\mB$ buyers sampled i.i.d. from $F_B$, when all agents report truthfully.

\begin{lemma}\label{lem:anon-problem-many}
	Let $\mS$ and $\mB$ be any numbers and assume that a prior-independent mechanism $M$ for $\mS$ sellers and $\mB$ buyers is deterministic, IR,  truthful, weakly budget-balanced and anonymous. Let $X_3>X_2>X_1>X_0\geq 0$. 
	If there is positive GFT in $M$ for some value profile where $s^{(1)}=X_1$ and\footnote{We emphasize that if $\mS=1$, then the assumption is that $s^{(1)}=X_1$, and it is not assumed that any agent has value $X_3$ (and therefore $X_3$ is of no consequence.).} $s^{(2)}=s^{(3)}=\cdots=s^{(\mS)}=X_3$, and $b^{(1)}=X_2$ and $b^{(2)}=b^{(3)}=\ldots b^{(\mB)}=X_0$, then the following holds. For every $\varepsilon>0$ there exist two distributions $F_S$ and $F_B$ for which 
	\[M(\mS,\mB)<\varepsilon\cdot \opt(1,1),\]
	and additionally for these distributions $\opt(\mS',\mB')= \sbopt(\mS',\mB')$ for every $\mS',\mB'$. 
	Furthermore, if $\mS=1$, then the result holds for some $F_B$ and $F_S$ such that $F_B~\text{FSD}~F_S$.
\end{lemma}
\begin{proof}
	Assume w.l.o.g. that the bid $b^{(1)}=X_2$ is from buyer $1$ and that the bid $s^{(1)}$ is from seller $1$. For that profile to have positive GFT, trade must be between buyer $1$ and seller $1$. 
	Let $p_b$ be the price paid by buyer $1$ and $p_s$ be the price paid to seller $1$. 
	By BB and IR it holds that $X_1\leq p_s\leq p_b\leq X_2$.
	Truthfulness for our single-parameter domain is characterized by monotonicity \cite{Myerson81}.  
	So by truthfulness, if seller $1$'s value were $X_0$ she would still trade at price $p_s$, 
	while buyer $1$'s price might change to $p'_b$, but by BB and IR, $p'_b$ also satisfies 
	$X_1\leq p_s\leq p'_b\leq X_2$. Now, again by truthfulness, if buyer $1$'s value drops to $Y=(X_0+X_1)/2$,
	there must be no trade as $(X_0+X_1)/2<X_1 \leq p'_b$.
	We conclude that for the profile in which buyer $1$'s value is $b^{(1)}=(X_0+X_1)/2$, all sellers except seller~$1$ have value $s^{(2)}=s^{(3)}=\cdots=s^{(\mS)}=X_3$, and all other agents (seller $1$ and all buyers except buyer $1$) have value $s=b^{(2)}=b^{(3)}=\ldots b^{(\mB)}=X_0$, there is no trade in $M$. 
	We next show that this implies that for any $\mS$  sellers and $\mB$ buyers, for some distributions $F_B$ and $F_S$ with the desired properties, mechanism $M$ does not even beat $\varepsilon\cdot\opt(1,1)$.   
	To prove this result we use a weakness of $M$: 
	as it is prior-independent, it will have no trade in the above profile regardless of the distributions $F_B$ and $F_S$.
	Now, by anonymity, as $M$ has no trade in one profile in which precisely one buyer has value $Y$ and one seller has value $X_0$, it has no trade in any such profile (when swapping the value of $Y$ between that buyer and some other buyer, and similarly for the value $X_0$ between sellers if there is more then one).

	We consider the following distributions $F_S$ and $F_B$: a seller has value $X_0$ with probability $1$ when $\mS=1$, and with probability  
	$\tilde{\varepsilon}$  when $\mS>1$ (for some $\tilde{\varepsilon}>0$ to be set later), and otherwise has value~$X_3$,
	while a buyer has value $Y=(X_0+X_1)/2$ with probability $\tilde{\varepsilon}$ 
	and value $X_0$ otherwise.
	Clearly $\opt(\mS,\mB)= \sbopt(\mS,\mB)$ for any $\mS$ and $\mB$, as trading at the fixed price of $X_0$ (that is, for every realization of values, for every $i\le q$ offering that $b^{(i)}$ trade with $s^{(i)}$ for a price of $X_0$, and trading if both agents accept) is optimal.
	Moreover, for $\mS=1$ the buyer distribution $F_B$ FSD the seller distribution $F_S$, as in this case the seller distribution is a point mass at $X_0$.
	
	We first consider the case that $\mS=1$. For the distributions $F_S$ and $F_B$, we observe that 
	$OPT(1,1)= (Y-X_0)\cdot \tilde{\varepsilon}$, as the probability that the buyer has value $Y$ is $\tilde{\varepsilon}$, the seller has value $X_0$ for sure, and if there is such a positive-GFT trade, the GFT is $Y-X_0$.
	
	Now we use the weakness of $M$ that we have observed above: as there is only one seller of value $X_0$, the mechanism $M$ has no trade (and hence no GFT) if only one buyer has value $Y$.  
 	So $M$ has positive GFT only if more than one buyer has value $Y$, and the probability that at least two buyers have value $Y$ is at most $\mB^2\cdot \tilde{\varepsilon}^2$ by a simple union bound. If there is such a positive-GFT trade, the GFT is  $Y-X_0$.
	So,  $M(1,\mB)\leq  (Y-X_0)\cdot (\mB^2\cdot \tilde{\varepsilon}^2) $ and thus for small enough $\tilde{\varepsilon}$ satisfying
	$\tilde{\varepsilon}<\frac{\varepsilon}{\mB^2}$	
	we get  
	\[M(1,\mB)<\varepsilon\cdot \opt(1,1),\]
	as needed.  
	
	Having handled the case in which $\mS=1$, we next consider the case in which $\mS>1$. 
	For the distributions $F_S$ and $F_B$, we observe that 
	$OPT(1,1)= (Y-X_0)\cdot \tilde{\varepsilon}^2$ as the probability that the buyer has value $Y$ is $\tilde{\varepsilon}$, and this is also the probability that the seller has value $X_0$. If there is such a positive-GFT trade, the GFT is $Y-X_0$.
	
	Now we again use the weakness of $M$ that we have observed above, yet now, we use its implications for $\mS>1$ sellers:
	$M$ has positive GFT only if either more than one buyer has value $Y$ and at least one seller has value $X_0$, or more than one seller has value $X_0$ and at least one buyer has value $Y$.
	Using simple union bounds,
	the probability that at least two buyers have value $Y$ is at most $\mB^2\cdot \tilde{\varepsilon}^2$,
	the probability that at least one seller has value $X_0$ is at most $\mS \cdot\tilde{\varepsilon}$,
	the probability that at least two sellers have value $X_0$ is at most $\mS^2 \cdot \tilde{\varepsilon}^2$,
	and the probability that at least one buyer has value $Y$ is at most $\mB\cdot  \tilde{\varepsilon}$.
	If one of the two events happen the GFT is at most $\mS\cdot (Y-X_0)$.
	So,  $M(\mS,\mB)\leq  \mS\cdot (Y-X_0) \cdot ((\mB^2\cdot \tilde{\varepsilon}^2)\cdot(\mS\cdot\tilde{\varepsilon}) + (\mS^2 \cdot \tilde{\varepsilon}^2)\cdot (\mB\cdot  \tilde{\varepsilon})))  =  (Y-X_0) \cdot (\mS^2\cdot\mB\cdot(\mB+\mS)\cdot \tilde{\varepsilon}^3)$. 

	We conclude that
	\[M(\mS,\mB)\leq (Y-X_0) \cdot (\mS^2\cdot\mB\cdot(\mB+\mS)\cdot \tilde{\varepsilon}^3).\]
	Recall that we have also observed  that
	\[OPT(1,1)= (Y-X_0)\cdot \tilde{\varepsilon}^2.\]
	Thus, for small enough $\tilde{\varepsilon}$ satisfying
	\[\tilde{\varepsilon}<\frac{\varepsilon}{\mS^2\cdot \mB\cdot(\mB+\mS)}\]
	we get  
		\[M(\mS,\mB)<\varepsilon\cdot \opt(1,1),\]
	as needed. 
\end{proof}

\begin{lemma}\label{lem:anon-problem}
	Let $\mB$ be any number and assume that a prior-independent mechanism $M$ for one seller and $\mB$ buyers is deterministic, IR,  truthful, weakly budget-balanced and anonymous.
	Let $X_2>X_1>X_0\geq 0$. 
	If for any setting with one seller with value supported in $\{X_0,X_1\}$ and $\mB$ buyers each with value supported in $\{X_0,X_2\}$ for some profile of values the mechanism $M$ has higher GFT than BTR, then the following holds. For every $\varepsilon>0$ there exist two distributions $F_S$ and $F_B$ that FSD $F_S$ for which 
	\[M(1,\mB)<\varepsilon\cdot \opt(1,1),\] and additionally for these distributions $\opt(1,\mB')= \sbopt(1,\mB')$ for every $\mB'$.
\end{lemma}
\begin{proof}
	For any setting with one seller with value in $\{X_0,X_1\}$  and $\mB$ buyers each with value in $\{X_0,X_2\}$, 
	the only profile in which BTR does not have optimal GFT is the profile in which the seller has value $s=X_1$, the highest buyer has value $b^{(1)}=X_2$, and all other buyers have value $X_0$, that is $b^{(2)}=b^{(3)}=\ldots b^{(\mB)}=X_0$. The lemma then follows from \cref{lem:anon-problem-many}.
\end{proof}

\section{Missing proofs from Section \ref{sec:iid}}\label{app:iid}

\subsection{Proof of Theorem~\ref{thm:iid}}

In this section we restate and prove \cref{thm:iid}.

\begin{theorem}[\cref{thm:iid}]
	Consider any setting with $\mS$ sellers and $\mB$ buyers whose values are sampled identically and independently from some distribution $F$ (that is, $F=F_S=F_B$).
	Then the Buyer Trade Reduction (BTR) mechanism with one additional buyer has expected GFT that is at least as large as the expected optimal GFT before adding the extra buyer. \\
	That is, for any setting with i.i.d.\ buyers and sellers it holds that  \[\btr(\mS,\mB+1)\geq \opt(\mS,\mB).\]
\end{theorem}

We stress that \cref{thm:iid} places no assumptions whatsoever on the distribution $F$, and in particular, does \emph{not} assume the regularity condition. This is in contrast to the celebrated result of \citeN{BK} that shows that for buyers with values sampled i.i.d.\ from a \emph{regular} distribution and a seller with no cost,  
the \emph{revenue} of a second-price auction (with no reserve) with one additional buyer independently sampled from the same distribution is at least as high the revenue of the optimal \citeN{Myerson81} mechanism before adding that buyer. That result indeed does not hold for some distributions that are not regular, while our result does. While the results are analogous (for different settings and different objectives, of course), the proofs of the two results are very different and do not seem to be related.

Before proving \cref{thm:iid}, we will prove a lemma that will be useful not only for the proof of that theorem, but also in additional proofs in this paper. The lemma essentially says (and when there are no ties, precisely says) that the only realizations for which BTR does not achieve the optimal welfare (equivalently, does not achieve the optimal GFT), are those in which a seller has the $(\mS\!+\!1)^{\text{st}}$ realized highest value. 

\begin{lemma}[\cref{lem:reduce}]
	In a double-auction  market with $\mS$ sellers and $\mB$ buyers, the Buyer Trade Reduction (BTR) mechanism achieves optimal GFT if and only if some buyer's value equals $x^{(\mS+1)}$, the $(\mS\!+\!1)^{\text{st}}$ realized highest value.\footnote{We emphasize that in the case of ties between values, the condition is that one of the agents whose realized values equal the $(\mS\!+\!1)^{\text{st}}$ realized highest value is a buyer.} Furthermore, if BTR does not achieve optimal GFT, then the efficient trade size $q$ is positive, and the GFT of BTR are lower than optimal by $b^{(q)}-x^{(\mS+1)}$.
\end{lemma}

\begin{proof}
	Recall that $q= q(\vecs, \vecb)$ denotes the efficient trade size for realization $(\vecs, \vecb)$ when ties are broken in favor of buyers. 
	If $q=0$ then by definition all seller values are strictly greater than all buyer values, so the highest buyer value equals $x^{(\mS+1)}$, and we are done with this case since when $q=0$, BTR achieves optimal gains from trade. We will assume henceforth that $q>0$.
	We will show that if $q>0$ then if BTR achieves optimal gains from trade then some buyer's value equals $x^{(\mS+1)}$, and if BTR does not achieve optimal gains from trade then no buyer's value equals $x^{(\mS+1)}$ and the loss is $b^{(q)}-x^{(\mS+1)}$.

	By definition of the efficient trade size $q$, the agents with the lowest $\mB$ values (with ties broken in favor of having as many sellers as possible among these $\mB$ agents) contain $q$ sellers: $s^{(q)} \geq \cdots \geq s^{(1)}$, and the agents with the top $\mS$ values (with ties broken in favor of having as many buyers as possible among these $\mS$ agents) contain $\mS-q$ sellers and $q$ buyers: $b^{(1)} \geq \cdots \geq b^{(q)}$.  Then $b^{(q+1)}$ has the highest value of any buyer among the $\mB$ agents with the lowest values, and $s^{(q)}$ has the highest value of any seller among these $\mB$ agents with lowest values. Therefore, $x^{(\mS+1)}=\max\{b^{(q+1)},s^{(q)}\}$.
	
	BTR offers a price of $b^{(q+1)}$ to all $b^{(i)}$ and $s^{(i)}$ for $1 \leq i \leq q$, and no reduction occurs if and only if $b^{(i)} \geq b^{(q+1)} \geq s^{(i)}$ for all $i$. This occurs precisely when $b^{(q+1)} \geq s^{(q)}$.  So, a trade is reduced if and only if $s^{(q)}>b^{(q+1)}$, in which case also $s^{(q)}=\max\{b^{(q+1)},s^{(q)}\}=x^{(\mS+1)}$. In this case, the reduced trade is of $s^{(q)}$ with $b^{(q)}$, and the gains of this trade are $b^{(q)}-s^{(q)}=b^{(q)}-x^{(\mS+1)}$.
	
	By the above, if BTR does not achieve optimal gains from trade, then a trade with positive gains $b^{(q)}-x^{(\mS+1)}$ is reduced, and therefore and by the above, $b^{(q)}>x^{(\mS+1)}>b^{(q+1)}$. In this case, therefore no buyer's value equals $x^{(\mS+1)}$, proving the first direction of the claim. Conversely, if BTR does achieve optimal gains from trade, then either no trade is reduced, or a trade with zero gains is reduced. In the former case, $b^{(q+1)} \geq s^{(q)}$ and therefore $b^{(q+1)}=\max\{b^{(q+1)},s^{(q)}\}=x^{(\mS+1)}$. In the latter case, $b^{(q)}=s^{(q)}=x^{(\mS+1)}$. So in either case there is a buyer whose value equals $x^{(\mS+1)}$, proving the second direction of the claim.
\end{proof}

\begin{proof}[Proof of \cref{thm:iid}]
	We prove the claim by comparing each side of the inequality to the optimal gains from trade in the augmented market with $\mS$ sellers and $\mB+1$ buyers, and showing that $\opt(\mS,\mB+1) - \opt(\mS,\mB) \geq \opt(\mS,\mB +1) - \btr(\mS,\mB+1)$, which clearly implies the claim.
	
	We will compare the above expectations by coupling them within a single probability space as follows: we first draw $\mS+\mB+1$ values for the ``augmented'' market from $F$, denoting the ordered values by $\ZZZ{1}{\mS}{\mB+1}\ge\ZZZ{2}{\mS}{\mB+1}\ge\cdots\ge\ZZZ{\mS+\mB+1}{\mS}{\mB+1}$. We then uniformly choose $\mS$ indices $I\subset\{1,\ldots,\mS+\mB+1\}$ and label the values $\ZZZ{i}{\mS}{\mB+1}$ for all $i\in I$ as sellers, denoting the remaining indices by $J$, and labeling the corresponding values as buyers. We then finally uniformly choose one index $j^*\in J$ to be labeled as the additional buyer who is not part of the market with only $\mB$ buyers.
	We denote the ordered values in the ``original'' market (that is, with $\ZZZ{j^*}{\mS}{\mB+1}$ removed) by $\ZZZ{1}{\mS}{\mB}\ge\ZZZ{2}{\mS}{\mB}\ge\cdots\ge\ZZZ{\mS+\mB}{\mS}{\mB}$.
	Note that this sampling process indeed generates the same probability spaces as sampling the values of the agents i.i.d. from $F$ in both the original market and the augmented market.  
	
	We will prove the desired inequality for every realization of $\ZZZ{1}{\mS}{\mB+1}\ge\ZZZ{2}{\mS}{\mB+1}\ge\cdots\ge\ZZZ{\mS+\mB+1}{\mS}{\mB+1}$, in expectation over the draws of seller/buyer indices and of the additional buyer index. So, for the remainder of this proof let us fix a realized value vector $\ZZZ{1}{\mS}{\mB+1}\ge\ZZZ{2}{\mS}{\mB+1}\ge\cdots\ge\ZZZ{\mS+\mB+1}{\mS}{\mB+1}$, and whenever we write an expectation either explicitly or implicitly (e.g., implicitly as in $\opt(\mS,\mB)$), let us mean an expectation over the draws of seller/buyer indices and additional buyer index, but not over the realized value vector, which we already fixed.
	
	First, we compare the expected optimum gains from trade in the augmented market, with $\mS$ sellers and $\mB+1$ buyers, to those in the original market, with only $\mS$ sellers and $\mB$ buyers, that is, we compute $\opt(\mS,\mB+1) - \opt(\mS,\mB)$.
	We can write the optimum gains from trade as 
	\[\opt(\mS, \mB+1)= \E\left[\sum_{j=1}^{\mS} \ZZZ{j}{\mS}{\mB+1} - \sum_{i\in I} \ZZZ{i}{\mS}{\mB+1}\right] \]
	and
	\[\opt(\mS, \mB)= \E\left[\sum_{j=1}^{\mS} \ZZZ{j}{\mS}{\mB} - \sum_{i\in I} \ZZZ{i}{\mS}{\mB+1}\right].\]
	We note that the second terms in the expectations, i.e., the welfares before trade, are identical in both markets.  
	
	In the event that the additional buyer is one of the $\mB+1$ lower-indexed agents, that is, $j^*>\mS$, the optimal GFT with and without the additional buyer are the same.
	On the other hand, we note that without conditioning on the choice of seller/buyer indices, $j^*$ is uniformly distributed in $\{1,\ldots,\mS+\mB+1\}$, and thus
	with probability $\frac{\mS}{\mS+\mB+1}$ the additional buyer is one of the $\mS$ higher-indexed values, that is, $j^*\le\mS$.  In this case, the difference between the optimal GFT in the augmented market and the original market is the additional buyer's value compared to the value of the $\mS$\textsuperscript{th} highest-valued agent in the original market $\ZZZ{\mS}{\mS}{\mB}$, which in this case is $\ZZZ{\mS+1}{\mS}{\mB+1}$.  
	Hence 
	\begin{multline*}
	\opt\left(\mS,\mB+1 ~\middle|~ j^*\le\mS\right) - \opt\left(\mS,\mB ~\middle|~ j^*\le\mS\right) =\\
	= \E\left[\ZZZ{j^*}{\mS}{\mB+1} - \ZZZ{\mS+1}{\mS}{\mB+1} ~\middle|~ j^*\le\mS\right] =\\
	= \E\left[\ZZZ{j^*}{\mS}{\mB+1} ~\middle|~ j^*\le\mS\right] - \ZZZ{\mS+1}{\mS}{\mB+1} =\\
	=\E_{i\in U\{1,\ldots,\mS\}}\left[\ZZZ{i}{\mS}{\mB+1}\right] - \ZZZ{\mS+1}{\mS}{\mB+1}.
	\end{multline*}
	
	Now we compare Buyer Trade Reduction on $\mS$ sellers and $\mB+1$ buyers with the optimal expected gains from trade on the same number of sellers and buyers.  
	By \cref{lem:reduce}, BTR is optimal if and only if some buyer's value equals $\ZZZ{j}{\mS}{\mB+1}$.
	So, a necessary condition for BTR not to be optimal (this condition is also sufficient unless there are ties between values) is for $\mS+1$ to be chosen as an index of a seller, which occurs with probability $\frac{\mS}{\mS+\mB+1}$ since we draw the seller indices uniformly.
	Recall that we use $q$ to denote the efficient trade size for the given realization. 
	In the event that $\mS+1$ is chosen as a seller index (and hence $q>0$, and hence $b^{(q)}$ is well defined and $b^{(q)}\ge\ZZZ{\mS+1}{\mS}{\mB+1}$), by \cref{lem:reduce} we have that\footnote{The inequality may only be strict in certain cases with ties between values.}
	\begin{multline*}
	\opt\left(\mS,\mB+1 ~\middle|~ \mS+1 \in I\right) - \btr\left(\mS,\mB+1 ~\middle|~ \mS+1 \in I\right) =\\
	= \E\left[\left(b^{(q)}- \ZZZ{\mS+1}{\mS}{\mB+1}\right)\cdot\mathbbm{1}[\text{BTR is not optimal}] ~\middle|~ \mS+1 \in I \right] \le\\
	\le \E\left[b^{(q)}- \ZZZ{\mS+1}{\mS}{\mB+1} ~\middle|~ \mS+1 \in I \right] =\\
	= \E\left[b^{(q)} ~\middle|~ \mS+1 \in I \right] - \ZZZ{\mS+1}{\mS}{\mB+1} =\\
	= \E\left[\ZZZ{\max(J\cap\{1,\ldots,\mS\})}{\mS}{\mB+1} ~\middle|~ \mS+1 \in I\right] - \ZZZ{\mS+1}{\mS}{\mB+1}.
	\end{multline*}
	
	To summarize,
	we find that 
	\[
	\opt(\mS,\mB+1) - \opt(\mS,\mB) = \frac{\mS}{\mS+\mB+1} \left(\E_{i\in U\{1,\ldots,\mS\}}\left[\ZZZ{i}{\mS}{\mB+1}\right] - \ZZZ{\mS+1}{\mS}{\mB+1}\right)
	\]
	and
	\begin{multline*}
	\opt(\mS,\mB+1) - \btr(\mS,\mB+1) \le \\
	\le \frac{\mS}{\mS+\mB+1} \left(\E\left[\ZZZ{\max(J\cap\{1,\ldots,\mS\})}{\mS}{\mB+1} ~\middle|~ \mS+1 \in I\right] - \ZZZ{\mS+1}{\mS}{\mB+1}\right).
	\end{multline*}
	
	The expectation in the first equation involves uniformly drawing an index $i$ between $1$ and $\mS$ and taking $\ZZZ{i}{\mS}{\mB+1}$. In the expectation in the second equation, the random variable $q=|J\cap\{1,\ldots,\mS\}|$ is not fixed but is always positive (since $\mS+1\in I$), and the expectation involves uniformly drawing $q$ indices between $1$ and $\mS$ and taking $\ZZZ{i}{\mS}{\mB+1}$, where $i$ is the \emph{highest} drawn index. Since the values are indexed in decreasing order, the second expectation no higher than the first (this is immediate if we additionally condition on any fixed positive $q$, and therefore holds even without such conditioning). Therefore,
	\[\opt(\mS,\mB+1) - \opt(\mS,\mB) \geq \opt(\mS,\mB+1) - \btr(\mS,\mB+1).\tag*{\qedhere}\]
\end{proof}

\subsection{Missing proofs: Lower bounds without the FSD assumption}\label{app:reg-not-suff}
\subsubsection{Proof of  Proposition~\ref{obs-reg-not-suff}}

We prove that for BTR, even for regular distributions there is no upper bound on the number of agents that we need to add for BTR to beat the optimum in the absence of any assumptions on the relation of the distributions $F_S$ and $F_B$. 
Formally we show that for any $\mS$, $\mB$, $\ell$, and $k$, there exist two regular distributions $F_S$ and $F_B$ (that does \emph{not} FSD $F_S$) such that when $\mS$ sellers are sampled i.i.d.\ from $F_S$ and $\mB$ buyers are sampled i.i.d. from $F_B$, it holds that adding $k$
buyers (sampled i.i.d.\ from $F_B$) as well as $\ell$ sellers (samples i.i.d.\ from $F_S$) is not enough for BTR to beat the optimal feasible mechanism (and thus also the optimum). 
The claim we prove next implies \cref{obs-reg-not-suff}, showing that BTR fails even if we compare to the optimum with only one seller and one buyer. 

\begin{proposition}\label{obs-reg-not-suff-app}
	For every $\varepsilon>0$ and every positive integers $\mS,\mB$, $\ell$, and $k$,
	there exist two regular distributions $F_S$ and $F_B$ such that
	\begin{multline*}
	\btr(\mS+\ell,\mB+k)
	<\varepsilon\cdot \sbopt(1,1) \le \\
	\leq\varepsilon\cdot \sbopt(\mS,\mB)= \varepsilon\cdot \opt(\mS,\mB).
	\end{multline*}
\end{proposition}
\begin{proof}
	For $\tilde{\varepsilon}>0$ we define the two distributions as follows.  
	The buyer distribution $F_B$ is uniform on $[0, \nicefrac{1}{\tilde{\varepsilon}}]$ while the seller $F_S$ distribution is a point mass at $\nicefrac{1}{\tilde{\varepsilon}}-2$, so sellers have this value for sure.
	We will later fix $\tilde{\varepsilon}$ to be some function of $\mS$, $\mB$, $\ell$, $k$, and $\varepsilon$,
	to make the claim hold for these two regular distributions.   
	
	We first observe that the realized optimal GFT can be obtained by the mechanism (that depends on the distribution) that simply posts a price of  $\nicefrac{1}{\tilde{\varepsilon}}-2$ and in which the size of trade is the maximal number of possible trades at this price (which is the minimum between $\mS$ and the number of buyers with value of at least  $\nicefrac{1}{\tilde{\varepsilon}}-2$). 
	Trade happens between this number of sellers, and the same number of buyers with the highest values (each of which has a value of at least  $\nicefrac{1}{\tilde{\varepsilon}}-2$). This is clearly the optimal feasible mechanism, as it obtains the optimal realized GFT, and is feasible (and in particular budget-balanced).
	Thus the bound for the GFT of BTR that we will prove with respect to the expected optimal GFT also holds with respect to the optimal feasible mechanism.
	
	We first observe that probability that a buyer has value at least $\nicefrac{1}{\tilde{\varepsilon}}-1$ is $\tilde{\varepsilon}$, and since any seller has value 
	$\nicefrac{1}{\tilde{\varepsilon}}-2$, when such a buyer trades the GFT from that trade is at least 1, thus $\opt(1,1)\geq 1\cdot \tilde{\varepsilon}= \tilde{\varepsilon}$. 
	
	We next look at BTR with $\mS+\ell$ sellers and $\mB+k$ buyers. 
	We observe that the probability that a buyer has value at least $\nicefrac{1}{\tilde{\varepsilon}}-2$ is $2\tilde{\varepsilon}$, and when such a buyer trades the GFT from that trade is at most 2.
	It therefore holds that $\btr(\mS+\ell,\mB + k)\leq 2\cdot (\mS+\ell) \cdot (2\tilde{\varepsilon})^2 \cdot (\mB+k)^2=8\cdot (\mS+\ell) \cdot \tilde{\varepsilon}^2 \cdot (\mB+k)^2$ since by a simple union bound, the probability that there are at least two buyers out of $\mB+k$ with values of at least $\nicefrac{1}{\tilde{\varepsilon}}-2$ is no more than $(2\tilde{\varepsilon})^2 \cdot (\mB+k)^2$.
	
	We conclude that  $\varepsilon\cdot \opt(1, 1) - \btr(\mS+\ell,\mB + k) \geq \varepsilon\cdot \tilde{\varepsilon} - 8\cdot (\mS+\ell) \cdot \tilde{\varepsilon}^2 \cdot (\mB+k)^2$, which is positive if
	$\tilde{\varepsilon} < \frac{\varepsilon}{8\cdot(\mS+\ell)\cdot(\mB+k)^2}$, and the claim follows.
\end{proof}

\subsubsection{Proof of  Proposition~\ref{obs-all-mech-reg-not-suff}}

To prove  \cref{obs-all-mech-reg-not-suff}, we first reprove \cref{obs-reg-not-suff-app} for a different pair of (in fact, irregular\footnote{Therefore, solely the statement of \cref{prop-FSD-nec-app} is strictly weaker than that of \cref{obs-all-mech-reg-not-suff}, however some features of the proof of \cref{obs-reg-not-suff-app}, which we will reuse later, in fact provide greater flexibility.}) distributions, and then use \cref{lem:anon-problem-many} to leverage properties of these particular distributions to generalize the claim for any anonymous robust deterministic mechanism.

\begin{proposition} \label{prop-FSD-nec-app}
	For every $\varepsilon>0$ and every positive integers $\mS$, $\mB$, $\ell$, and $k$,
	there exist two distributions $F_S$ and $F_B$ such that
	\begin{multline*}
	\btr(\mS+\ell,\mB+k)
	<\varepsilon\cdot \sbopt(1,1)\le\\
	\leq\varepsilon\cdot \sbopt(\mS,\mB)= \varepsilon\cdot \opt(\mS,\mB).
	\end{multline*}
\end{proposition}
\begin{proof}
	Consider the following distributions. 
	Each buyer value is sampled from the distribution $F_B$: the buyer has value $2$ with probability $\tilde{\varepsilon}$ (to be chosen later), and value $0$ otherwise. 
	Each seller value is sampled from the distribution $F_S$: the seller has value $1$ with probability $\tilde{\varepsilon}$, and value $3$ otherwise.  
	
	We first observe that for any number of buyers and any number of sellers, the realized optimum GFT can be obtained by the mechanism (that depends on the distributions) that simply posts a price of 1.5 and in which the size of trade is the maximal number of possible trades at this price (which is the minimum between the number of sellers with value $1$ and the number of buyers with value $2$.) 
	Trade happens between this number of buyers with value $2$ and this number of sellers with value $1$.
	This is clearly the optimal feasible mechanism, as it obtains the optimal GFT, and does so with a strongly budget-balanced mechanism.  
	Thus, the bound for the GFT of BTR that we will prove with respect to the expected optimal GFT also holds with respect to the optimal feasible mechanism.
	We conclude that for this distribution $ \sbopt(1,1)=  \opt(1,1)\leq  \sbopt(\mS,\mB)= \opt(\mS,\mB)$.
	
	Clearly $\opt(1,1)=\tilde{\varepsilon}^2$ as there is a GFT of $1$ 
	in the case that one buyer has value of $2$ and one seller has value of $1$, which happens with probability $\tilde{\varepsilon}^2$. 
	
	Next consider running BTR with $\mS+\ell$ sellers and $\mB+k$ buyers. 
	For this mechanism to have any positive GFT for a realization, it must be the case that at least one seller out of $\mS+\ell$ has value~$1$, and at least two buyers out of $\mB+k$ have values of $2$. In any such case the GFT is at most $\mS+\ell$. 
	Thus by a simple union bound for each of these probabilities we have that $\btr(\mS+\ell,\mB + k)\leq (\mS+\ell) \cdot ((\mS+\ell)\cdot\tilde{\varepsilon}) \cdot ((\mB+k)^2\cdot\tilde{\varepsilon}^2)=(\mS+\ell)^2\cdot(\mB+k)^2\cdot\tilde{\varepsilon}^3$, and so
	for $\tilde{\varepsilon} < \frac{\varepsilon}{(\mS+\ell)^2\cdot(\mB+k)^2}$ we have that $\btr(\mS+\ell,\mB+k)<\varepsilon\cdot\opt(1,1)$.
\end{proof}

We next strengthen \cref{prop-FSD-nec-app}, and show that the same negative result holds for \emph{any} anonymous robust deterministic mechanism. 
The claim we prove next implies \cref{obs-all-mech-reg-not-suff}, showing that the mechanism fails even if we compare to the optimum with only one seller and one buyer. 
The proof uses the distributions presented in \cref{prop-FSD-nec-app}, together with \cref{lem:anon-problem-many}.

\begin{proposition}
	For any prior-independent mechanism  $M$ that is deterministic, IR, truthful, weakly budget-balanced, and anonymous the following holds. 
	For every $\varepsilon>0$ and every positive integers $\mS$, $\mB$, $\ell$, and $k$,
	there exist two distributions $F_S$ and $F_B$ such that
	\[
	M(\mS+\ell,\mB+k)
	<\varepsilon\cdot \sbopt(1,1)
	\leq\varepsilon\cdot \sbopt(\mS,\mB)= \varepsilon\cdot \opt(\mS,\mB).
	\]
\end{proposition}
\begin{proof}
	We consider the distributions used in the proof of \cref{prop-FSD-nec-app}.
	Each buyer value is sampled from the distribution $F_B$: the buyer has value $2$ with probability $\tilde{\varepsilon}$ (to be chosen later), and value $0$ otherwise. 
	Each seller value is sampled from the distribution $F_S$: the seller has value $1$ with probability $\tilde{\varepsilon}$, and value $3$ otherwise.  
	
	As already shown in the proof of that proposition, $ \sbopt(1,1)=  \opt(1,1)\leq\linebreak\sbopt(\mS,\mB)= \opt(\mS,\mB)$ and $\opt(1,1)=\tilde{\varepsilon}^2$. 
	Consider any mechanism $M'$ that runs on $\mS+\ell$ sellers and $\mB+k$ buyers, and that 
	has optimal GFT on all value profiles except for profiles in which exactly one buyer has value $2$ and exactly one seller has value $1$. 
	We therefore have that mechanism $M'$ has GFT of at most $(\mS+\ell)\cdot({(\mS+\ell)\cdot(\mB+k)^2\cdot\tilde{\varepsilon}^3}+{(\mS+\ell)^2\cdot(\mB+k)\cdot\tilde{\varepsilon}^3})=(\mS+\ell)^2\cdot(\mB+k)\cdot(\mS+\ell+\mB+k)\cdot\tilde{\varepsilon}^3$, as at least two buyers must have value $2$ and one seller must have value $1$, or alternatively at least two sellers must have value $1$ and one buyer must have value $2$ (we bound the probability of each of these events by a simple union bound), and whenever there is positive GFT, the GFT is $1$.
	Therefore, for $\tilde{\varepsilon}<\frac{\varepsilon}{(\mS+\ell)^2\cdot(\mB+k)\cdot(\mS+\ell+\mB+k)}$ we have that $M'(\mS+\ell,\mB+k)<\varepsilon\cdot \opt(1,1)$.
	Thus for the given mechanism $M$ that runs on $\mS+\ell$ sellers and $\mB+k$ buyers to get GFT large enough to beat $\varepsilon\cdot \opt(1,1)$ for $F_B$ and $F_S$, it must have positive GFT for at least one realization in which exactly one buyer has value $2$ and exactly one seller has value $1$. Now we can use \cref{lem:anon-problem-many} for $X_3=3>X_2=2>X_1=1>X_0=0$, $\mS+\ell$ sellers, and $\mB+k$ buyers, to derive the claim. 
\end{proof}

\section{Missing proofs from Section \ref{sec-11}}\label{app:11}

\subsection{Proof of Proposition~\ref{prop:sd-1-1-lb}}

In this section we restate and prove \cref{prop:sd-1-1-lb}.
 
\begin{proposition} [\cref{prop:sd-1-1-lb}]
	There exist two distributions, a seller distribution $F_S$ and a buyer distribution $F_B$, such that $F_B~\text{FSD}~F_S$ and for which \[\btr(1,2)<\sbopt(1,1)= \opt(1,1).\]
\end{proposition}
\begin{proof}
	We consider the following distributions for the buyer value and seller value, where $\varepsilon\in(0,1)$ will be chosen later in this proof (for any such $\varepsilon$, the former distribution FSD the latter).

	\[b = \begin{cases} 0 & w.p.\,\, \vareps \\ 2 & w.p.\,\, 1-\vareps \end{cases} \quad\quad s = \begin{cases}0 & w.p.\,\, \vareps \\ 1 & w.p.\,\, 1-\vareps  \end{cases}\]
	
	We first observe that the optimal GFT can be obtained by the feasible mechanism (that depends on the distribution) that simply posts a price of 1.5 and trades the item whenever both agents agree to this price. So this is the optimal feasible mechanism and indeed $\sbopt(1,1)= \opt(1,1)$, and
	thus the bound for the GFT of BTR that we will prove with respect to the optimal GFT also holds with respect to the optimal feasible mechanism.
	
	We observe that \[\opt(1,1)= 2\cdot \varepsilon \cdot (1-\varepsilon) + 1\cdot (1-\varepsilon)^2 = (1-\varepsilon)\cdot (1+\varepsilon),\]  
	while 
	\[\btr(1,2) = \underbrace{(1-\varepsilon)^2}_{b_1=b_2=2}\cdot(2\cdot \varepsilon + 1\cdot (1-\varepsilon) ) + \underbrace{2\cdot\varepsilon\cdot (1-\varepsilon)}_{\{b_1,b_2\}=\{0,2\}} \cdot 2\cdot \varepsilon = (1-\varepsilon)\cdot(1+3\cdot\varepsilon^2).\]
	So, taking any $\varepsilon\in (0,\frac{1}{3})$, the claim follows as $3\cdot\varepsilon^2<\varepsilon$ for such $\varepsilon$.
\end{proof}

\subsection{Proof of Lemma~\ref{obs-all-mech-reg-not-suff-12}}

In this section we restate and prove \cref{obs-all-mech-reg-not-suff-12}.

\begin{lemma}[\cref{obs-all-mech-reg-not-suff-12}]
	For any prior-independent mechanism $M$ for the setting of one seller and two buyers,  
	that is deterministic, IR, truthful, weakly budget-balanced, and anonymous the following holds. 
	If there is any value profile for which the welfare of $M$ is higher than the welfare of $BTR$, then 
	for every $\varepsilon>0$ there exist two distributions $F_S$ and $F_B$ such that $F_B~\text{FSD}~F_S$ and for which
	\[M(1,2)<\varepsilon\cdot \sbopt(1,1)= \varepsilon\cdot \opt(1,1).\]
\end{lemma}
\begin{proof}
	For any profile with the seller value being no larger than both of the buyers values, BTR has optimal GFT. So the only cases in which $M$ can have higher GFT than BTR is if $b_1>s>b_2$ or $b_2>s>b_1$ and there is trade in $M$ between the high-value buyer and the seller. Assume that this happen for $b_2>s>b_1$ (the proof for the other case is analogous). We now observe that this setting satisfies the conditions of \cref{lem:anon-problem-many} from \cref{app:lower-bounds} for $X_2=b_2>X_1=s>X_0=b_1\geq 0$. Thus by the lemma 
	there exist two distributions $F_S$ and $F_B$ such that $F_B~\text{FSD}~F_S$ and for which 
	\[\varepsilon\cdot \opt(1,1) = \varepsilon\cdot \sbopt(1,1)> M(1,2)\] as needed.
\end{proof}

\subsection{Proof of Theorem~\ref{cor:iid-sample-half}}

\begin{theorem}[\cref{cor:iid-sample-half}]
	For a single buyer and a single seller with values drawn i.i.d.\ from $F=F_S=F_B$, pricing at a single independent sample drawn from $F$ obtains expected GFT at least $\frac{1}{2} \opt(1,1)$. Furthermore, if $F$ is atomless, then the obtained expected GFT is exactly $\frac{1}{2} \opt(1,1)$ (and no more than that).
\end{theorem}

\begin{proof}
As discussed in \cref{sec-11}, the first part follows directly from \cref{obs:one-sample,thm:iid}. We will nonetheless give a direct proof of both parts of the \lcnamecref{cor:iid-sample-half} at once.

Denote the expected GFT in bilateral trade when pricing at a sample drawn independently from $F$ by $\textsc{Sample}$.
The realized optimal GFT is $b-s$ whenever $b > s$ (and it is $0$ otherwise).
So,
\[\opt(1,1) = \E\left[(b-s) \cdot \mathbbm{1}[b\ge p\ge s]\right] + \E\left[(b-s) \cdot \mathbbm{1}[p>b>s]\right] + \E\left[(b-s) \cdot \mathbbm{1}[b>s>p]\right].\]
On the other hand, the realized GFT of pricing at the sample is $b-s$ \emph{only} in the event that the sample price falls weakly below the buyer and weakly above the seller (and it is $0$ otherwise).
So,
\[\textsc{Sample} =  \E\left[(b-s) \cdot \mathbbm{1}[b\ge p\ge s]\right].\]
To prove both parts of the theorem, it therefore suffices to prove that 
\[\textsc{Sample} \ge \E\left[(b-s) \cdot \mathbbm{1}[p>b>s]\right] + \E\left[(b-s) \cdot \mathbbm{1}[b>s>p]\right],\]
with equality for atomless distributions. To see this, consider the following equivalent process for drawing $b$, $s$, and $p$. We first draw three values independently and uniformly from $F$. We call these values, after ordering them, $x^{(1)}\ge x^{(2)}\ge x^{(3)}$. After drawing the values, we draw a permutation $\{i,j,k\}$ of $\{1,2,3\}$ uniformly at random, and set $b=x^{(i)}$, $s=x^{(j)}$, and $p=x^{(k)}$. It is straightforward that this procedure is indeed equivalent to drawing $b,s,p\sim F$ independently. Our proof will use the fact that each permutation is equally likely, independently of the drawn values:
\begin{align*}
\textsc{Sample} &= \E\left[(b-s) \cdot \mathbbm{1}[b\ge p\ge s]\right] \ge \\
&\ge\E\left[(b-s) \cdot \mathbbm{1}[b>p>s]\right] = \\
&=\E\left[(x^{(1)}-x^{(3)}) \cdot \mathbbm{1}[x^{(1)}>x^{(2)}>x^{(3)} \And i=1 \And k=2 \And j=3]\right] = \\
&=\frac{1}{6}\cdot\E\left[(x^{(1)}-x^{(3)}) \cdot \mathbbm{1}[x^{(1)}>x^{(2)}>x^{(3)}]\right] = \\
&=\frac{1}{6}\cdot\left(\E\left[(x^{(1)}-x^{(2)}) \cdot \mathbbm{1}[x^{(1)}>x^{(2)}>x^{(3)}]\right]+\E\left[(x^{(2)}-x^{(3)}) \cdot \mathbbm{1}[x^{(1)}>x^{(2)}>x^{(3)}]\right]\right) = \\
&=\E\left[(x^{(1)}-x^{(2)}) \cdot \mathbbm{1}[x^{(1)}>x^{(2)}>x^{(3)} \And k=1 \And i=2 \And j=3]\right]+ \\* &\qquad\qquad+\E\left[(x^{(2)}-x^{(3)}) \cdot \mathbbm{1}[x^{(1)}>x^{(2)}>x^{(3)}] \And i=1 \And j=2 \And k=3\right] = \\
&= \E\left[(b-s) \cdot \mathbbm{1}[p>b>s]\right] + \E\left[(b-s) \cdot \mathbbm{1}[b>s>p]\right],
\end{align*}
where the second and fourth equalities are since any permutation of $i,j,k$ is equally likely, and
where the only inequality becomes an equality if $F$ is atomless, as required.
\end{proof}

\subsection{Proof of Theorem~\ref{thm:FSD-11-sample}}

In this section we restate and prove \cref{thm:FSD-11-sample}.

\begin{theorem} [\cref{thm:FSD-11-sample}]
	For a single buyer with value drawn from $F_B$ and a single seller with value drawn from $F_S$ where $F_B~\text{FSD}~F_S$, pricing at a single independent sample drawn from $F_B$ obtains expected GFT at least $\frac{1}{4} \opt(1,1)$.\end{theorem}

\begin{proof}
Denote the expected GFT in bilateral trade when pricing at a sample drawn independently from $F_B$ by $\textsc{Sample}$.
The realized optimal GFT is $b-s$ whenever $b > s$ (and it is $0$ otherwise).
So,
\[\opt(1,1) = \E\left[(b-s) \cdot \mathbbm{1}[b\ge p\ge s]\right] + \E\left[(b-s) \cdot \mathbbm{1}[p>b>s]\right] + \E\left[(b-s) \cdot \mathbbm{1}[b>s>p]\right].\]
On the other hand, the realized GFT of pricing at the sample is $b-s$ \emph{only} in the event that the sample price falls weakly below the buyer and weakly above the seller (and it is $0$ otherwise).
So,
\[\textsc{Sample} =  \E\left[(b-s) \cdot \mathbbm{1}[b\ge p\ge s]\right].\]

To prove the claim, we will prove the following inequalities:
\begin{align}
\E\left[(b-s) \cdot \mathbbm{1}[p>b>s]\right]&\le\E\left[(b-s) \cdot \mathbbm{1}[b\ge p\ge s]\right], \label{pbs}\\
\E\left[(b-s) \cdot \mathbbm{1}[b>s>p]\right]&\le\E\left[(b-s) \cdot \mathbbm{1}[b\ge p\ge s]\right] + \E\left[(b-s) \cdot \mathbbm{1}[p>b> s]\right]. \label{bsp}
\end{align}
Combining \cref{bsp,pbs} indeed gives that $\textsc{Sample}$ is a $\frac{1}{4}$-approximation to $\opt(1,1)$:
\begin{align*}
\opt(1,1) &= \E\left[(b-s) \cdot \mathbbm{1}[b\ge p\ge s]\right] + \E\left[(b-s) \cdot \mathbbm{1}[p>b>s]\right] + \E\left[(b-s) \cdot \mathbbm{1}[b>s>p]\right]\le\\
&\le 2\cdot\E\left[(b-s) \cdot \mathbbm{1}[b\ge p\ge s]\right] + 2\cdot\E\left[(b-s) \cdot \mathbbm{1}[p>b>s]\right]\le\tag*{By \cref{bsp}}\\
&\le 4\cdot\E\left[(b-s) \cdot \mathbbm{1}[b\ge p\ge s]\right] = 4\cdot\textsc{Sample}.\tag*{By \cref{pbs}}
\end{align*}

We will prove each of \cref{bsp,pbs} by considering the following equivalent process for drawing $b$, $s$, and $p$. We first draw three quantiles (see \cref{sec:prelim-quantiles}) independently and uniformly from $(0,1)$. We call these draws, after ordering them, $q^{(1)}\ge q^{(2)}\ge q^{(3)}$. For any quantile $q'\in(0,1)$, we denote $b^{q'}=v_{F_B}(q')$
and similarly $s^{q'} = v_{F_S}(q')$.
Since $F_B$ FSD $F_S$, we have that $b^{q'} = v_{F_B}(q') \geq v_{F_S}(q') = s^{q'}$ for any $q'$.
After drawing the quantiles, we draw a permutation $\{i,j,k\}$ of $\{1,2,3\}$ uniformly at random, and set $b=b^{q^{(i)}}$, $s=s^{q^{(j)}}$, and $p=b^{q^{(k)}}$.
As noted in \cref{sec:prelim-quantiles}, generating a value according to a quantile distributed uniformly in $(0,1)$ and then taking the value that corresponds to that quantile for some distribution results in a value distributed according to that distribution. As such, our procedure is indeed equivalent to drawing $s\sim F_S$ and $b,p\sim F_B$ independently. The proofs of both \cref{pbs,bsp} will use the fact that any permutation of $i,j,k$ is equally likely, independently of the drawn quantiles.

We start by proving \cref{pbs}:
\[
\E_{i,j,k}\left[(b-s) \cdot \mathbbm{1}[p>b>s]\right]\le\E_{i,j,k}\left[(p-s) \cdot \mathbbm{1}[b\ge p\ge s]\right]\le\E_{i,j,k}\left[(b-s) \cdot \mathbbm{1}[b\ge p\ge s]\right],
\]
where the first equality is since any permutation of $i,j,k$ is equally likely.

We conclude the proof by proving \cref{bsp}.  Recall that $b=b^{q^{(i)}}$, $s=s^{q^{(j)}}$, and $p=b^{q^{(k)}}$. We denote $p'=b^{q^{(j)}}$ and $s'=s^{q^{(k)}}$, corresponding to switching the quantiles of the price and the seller value. We note that whenever the value order is $b>s>p$, by stochastic dominance necessarily $k=3$ (recall that $k$ is defined such that $p=b^{q^{(k)}}$) and so $s'=s^{q^{(3)}}\le s$ and by stochastic dominance $s'\le p<s\le p'$ and $s'\le s< b$. Therefore:
\begin{align*}
\E_{i,j,k}\left[(b-s) \cdot \mathbbm{1}[b>s>p]\right] &\le\E_{i,j,k}\left[(b-s') \cdot \mathbbm{1}[b>s>p]\right]=\\
&= \E_{i,j,k}\left[(b-s') \cdot \mathbbm{1}[b>s>p \And b \ge p']\right] + \\*
&\qquad\qquad+\E_{i,j,k}\left[(b-s') \cdot \mathbbm{1}[b>s>p \And b < p']\right]\le\\
&\le \E_{i,j,k}\left[(b-s') \cdot \mathbbm{1}[b\ge p'\ge s']\right]+\E_{i,j,k}\left[(b-s') \cdot \mathbbm{1}[p'>b>s']\right]=\\
&= \E_{i,j,k}\left[(b-s) \cdot \mathbbm{1}[b\ge p\ge s]\right]+\E_{i,j,k}\left[(b-s) \cdot \mathbbm{1}[p>b>s]\right],
\end{align*}
where the last equality is since any permutation of $i,j,k$ is equally likely.
\end{proof}

\subsection{Proof of Proposition~\ref{prop:FSD-11-sample-nohalf}}

\begin{proposition} [\cref{prop:FSD-11-sample-nohalf}]
	For every $\alpha>\frac{7}{16}$ there exist two distributions, a seller distribution $F_S$ and a buyer distribution $F_B$, such that $F_B~\text{FSD}~F_S$ and for which for a single buyer with value drawn from $F_B$ and a single seller with value drawn from $F_S$, pricing at a single independent sample drawn from $F_B$ obtains expected GFT of less than $\alpha\cdot\sbopt(1,1)=\alpha\cdot\opt(1,1)$.
\end{proposition}

\begin{proof}
	We use a construction similar to the one from the proof of \cref{prop:sd-1-1-lb}, albeit smoothing out the atom at the high buyer value to avoid unnecessary revenue due to the fact that ties between the sample and buyer are by definition broken in favor of the buyer. (This is the same reason for which the converse of \cref{obs:one-sample} fails to hold with distributions with atoms, as discussed in \cref{footnote-sampling}.) We will furthermore optimize over the high buyer value.
	We therefore consider the following distributions for the buyer value and seller value, where $\delta>0$ should be thought of as small, $\gamma$ should be though of as large, and $\varepsilon\in(0,1)$ will be chosen later in this proof (for any such $\varepsilon$ and $\delta<\gamma$, the former distribution FSD the latter).

	\[b = \begin{cases} 0 & w.p.\,\, \vareps \\ U[1+\gamma-\delta,1+\gamma+\delta] & w.p.\,\, 1-\vareps \end{cases} \quad\quad s = \begin{cases}0 & w.p.\,\, \vareps \\ 1 & w.p.\,\, 1-\vareps  \end{cases}\]
	
	We first observe that the optimal GFT can be obtained by the feasible mechanism (that depends on the distribution) that simply posts a price of $1+\frac{\gamma-\delta}{2}$ and trades the item whenever both agents agree to this price. So this is the optimal feasible mechanism and indeed $\sbopt(1,1)= \opt(1,1)$, and
	thus the bound for the GFT of BTR that we will prove with respect to the optimal GFT also holds with respect to the optimal feasible mechanism.
	
	We observe that \[\opt(1,1)= (1+\gamma)\cdot \varepsilon \cdot (1-\varepsilon) + \gamma\cdot (1-\varepsilon)^2 = (1-\varepsilon)\cdot (\gamma+\varepsilon),\]  
	while pricing at a random sample from $F_B$ the expected revenue is:
	\begin{multline*}
	\underbrace{\frac{(1-\varepsilon)^2}{2}}_{b\ge p>0}\cdot\bigl((1+\gamma+\nicefrac{\delta}{3})\cdot \varepsilon + (\gamma+\nicefrac{\delta}{3})\cdot (1-\varepsilon)\bigr) + \underbrace{\varepsilon\cdot (1-\varepsilon)}_{b>p=0} \cdot (1+\gamma)\cdot \varepsilon = \\
	=\frac{1}{2}\cdot(1-\varepsilon)\cdot(\varepsilon+(1-\varepsilon)\cdot\gamma+(1+2\gamma)\cdot\varepsilon^2+(1-\varepsilon)\cdot\nicefrac{\delta}{3}) <
	\frac{1}{2}\cdot(1-\varepsilon)\cdot(\varepsilon+(1-\varepsilon)\cdot\gamma+(1+2\gamma)\cdot\varepsilon^2+\delta).
	\end{multline*}
	\noindent The ratio between the former and the latter is:
	\[2\cdot\frac{\gamma+\varepsilon}{\varepsilon+(1-\varepsilon)\cdot\gamma+(1+2\gamma)\cdot\varepsilon^2+\delta}.\]
	Taking $\varepsilon=\nicefrac{1}{4}$, the ratio becomes:
	\[2\cdot\frac{\gamma+\nicefrac{1}{4}}{\nicefrac{7}{8}\cdot\gamma+\nicefrac{5}{16}+\delta},\]
	which for any fixed $\delta$ can be made arbitrarily close to \nicefrac{16}{7} by taking $\gamma$ large enough, completing the proof.
\end{proof}

\subsection{Proof of Lemma~\ref{obs-all-mech-reg-not-suff-12}}

In this section we restate and prove \cref{obs-all-mech-reg-not-suff-12}.

\begin{lemma}[\cref{obs-all-mech-reg-not-suff-12}]
	For any prior-independent mechanism $M$ for the setting of one seller and two buyers,  
	that is deterministic, IR, truthful, weakly budget-balanced, and anonymous the following holds. 
	If there is any value profile for which the welfare of $M$ is higher than the welfare of $BTR$, then 
	for every $\varepsilon>0$ there exist two distributions $F_S$ and $F_B$ such that $F_B~\text{FSD}~F_S$ and for which
	\[M(1,2)<\varepsilon\cdot \sbopt(1,1)= \varepsilon\cdot \opt(1,1).\]
\end{lemma}
\begin{proof}
	For any profile with the seller value being no larger than both of the buyers values, BTR has optimal GFT. So the only cases in which $M$ can have higher GFT than BTR is if $b_1>s>b_2$ or $b_2>s>b_1$ and there is trade in $M$ between the high-value buyer and the seller. Assume that this happen for $b_2>s>b_1$ (the proof for the other case is analogous). We now observe that this setting satisfies the conditions of \cref{lem:anon-problem-many} from \cref{app:lower-bounds} for $X_2=b_2>X_1=s>X_0=b_1\geq 0$. Thus by the lemma 
	there exist two distributions $F_S$ and $F_B$ such that $F_B~\text{FSD}~F_S$ and for which 
	\[\varepsilon\cdot \opt(1,1) = \varepsilon\cdot \sbopt(1,1)> M(1,2)\] as needed.
\end{proof}

\subsection{Proof of Lemma~\ref{obs:k-samples}}

In this section we restate and prove \cref{obs:k-samples}.

\begin{lemma}[\cref{obs:k-samples}]
	For any seller distribution $F_S$ and any buyer distribution $F_B$, if the following holds:
	\begin{itemize}
		\item
		$\btr(1, 1+k) \geq \opt(1,1)$,
	\end{itemize}
	then the following also holds:
	\begin{itemize}
		\item
		In a setting with one seller and one buyer, pricing at the maximum of $k$ samples drawn independently from $F_B$ obtains expected GFT at least $\frac{1}{1+k} \opt(1,1)$.
	\end{itemize}
	Furthermore, if $F_B$ is atomless, then the converse implication is also true: if the latter statement holds, then so does the former.\footnote{Yet, for some distributions with atoms, the converse implication is actually false, see \cref{footnote-sampling}. 	}
\end{lemma}

\begin{proof}
	Denoting the expected GFT in bilateral trade when pricing at the maximum of $k$ samples drawn independently from $F_B$ by $\textsc{Sample}_k$, we calculate:
	\begin{align*}
	\btr(1,1+k) & = \E_{(b_{1}, \ldots, b_{k+1})\sim F_B^{k+1}, s \sim F_S} \left[ \sum _{i=1} ^{1+k} (b_i -s) \cdot \mathbbm{1} \left[i=\min\bigl\{j~\middle|~b_j\!=\!\max\{s,b_1,\ldots,b_{k+1}\}\bigr\}\right]\right] \le\\
	&\le \E_{(b_{1}, \ldots, b_{k+1})\sim F_B^{k+1}, s \sim F_S} \left[ \sum _{i=1} ^{1+k} (b_i -s) \cdot \mathbbm{1} \left[b_i=\max\{s,b_1,\ldots,b_{k+1}\}\right]\right] =\\
	&= (1+k)\cdot\E_{b \sim F_B, s \sim F_S, (p_1,\ldots,p_k)\sim F_B^k} \left[ (b-s) \cdot \mathbbm{1} \left[ b\ge s \And b \ge \max\{p_1,\ldots,p_k\} \right] \right] =\\
	&= (1+k)\cdot\textsc{Sample}_k,
	\end{align*}
	with the inequality becoming an equality is $F_B$ is atomless, since in that case with probability $1$ there are no ties between the drawn buyer values and so the expressions in the expectation in the two sides of the inequality are equal.
\end{proof}

\section{Missing proofs from Section \ref{sec:1r}}
\subsection{Proof of Theorem~\ref{thm:sd-1-r-lb}}
In this section we restate and prove \cref{thm:sd-1-r-lb}.

\begin{theorem}[\cref{thm:sd-1-r-lb}]
	For any prior-independent mechanism $M$ that is deterministic, IR, truthful, weakly budget-balanced, and anonymous the following holds. 
	For any positive number $\mB$ there exist two distribution $F_B$ and $F_S$ such that
	$F_B$ FSD $F_S$ and for which 
	\[M(1,\mB+\lfloor\log_2 \mB\rfloor)<\sbopt(1,\mB)= \opt(1,\mB).\]
\end{theorem}
\begin{proof}
	The claim is for $\mB=1$ clearly follows form the \citeN{MyersonS83} impossibility result, 
	so we will prove it for $\mB\ge2$.
	We first prove the claim for the BTR mechanism, that is, we show that there exist a seller distribution $F_S$ and a buyer distribution $F_B$ such that $F_B~\text{FSD}~F_S$ for which 
	\[\btr(1,\mB+\lfloor\log_2 \mB\rfloor)<\sbopt(1,\mB)= \opt(1,\mB)\]
	After proving this claim we will show that it implies the result for the given anonymous and robust deterministic mechanism~$M$. 
	
	To prove the claim for BTR we use the following distributions: 
	\[b = \begin{cases} 2 & w.p.\,\, 1/2 \\ 0 & w.p.\,\, 1/2 \end{cases} \quad\quad s = \begin{cases} 1 & w.p.\,\, 1/2 \\ 0 & w.p.\,\, 1/2 \end{cases}\]
	
	The optimal feasible mechanism (which depends on the distributions) can realize the optimal GFT by posting a price of 1.5. So the bound we prove with respect to the expected optimal GFT will also hold for the optimal feasible mechanism. 
	
	Let
	\[p = \Pr [\text{at least $1$ of $\mB$ buyers has value $2$}] = 1 - 2^{-\mB}\]
	be the probability that at least one out of $\mB$ buyer values sampled i.i.d.\ from the buyer values distribution is $2$.
	It holds that
	\[\opt(1,\mB) = \frac{1}{2}\cdot (2-1) \cdot p + \frac{1}{2}\cdot (2-0) \cdot p = \frac{3p}{2}.\]
	For any $k$, let
	\[q=  \Pr [\text{at least $2$ of $\mB+k$ buyers has value $2$}] = 1 - 2^{-(\mB+k)} - (\mB+k) 2^{-(\mB+k)}\]
	be the probability that at least two out of $\mB+k$ buyers values sampled i.i.d.\ from the buyer values distribution are $2$.
	It holds that
	\[\btr(1,\mB+k) = \frac{1}{2}\cdot (2-1) \cdot q + \frac{1}{2}\cdot (2-0) \cdot q = \frac{3q}{2}.\]
	
	Now, $\opt(1,\mB)>\btr(1,\mB+k)$ if and only if $p>q$, or equivalently,
	\[1 - 2^{-\mB} > 1 - 2^{-(\mB+k)} - (\mB+k) 2^{-(\mB+k)},\]
	which is equivalent to $2^k < \mB+k+1$.
	Specifically, for $k=\lfloor\log_2\mB\rfloor$ we have $2^k\le \mB < \mB+k+1$, and so $\btr(1,\mB+k)<\opt(1,\mB)$, as required.
	
	Now, we observe that the above distributions satisfy the condition of \cref{lem:anon-problem} from \cref{app:lower-bounds} (for $X_2=2>X_1=1>X_0=0$). By that lemma,
	if there is some profile of values in which the given mechanism $M$ has higher GFT than BTR, then there exist distributions $F'_S$ and $F'_B$ such that $F'_B~\text{FSD}~F'_S$ for which 
	$\opt(1,1) > M(1,\mB+\lfloor\log_2\mB\rfloor)$ and also 
	$\opt(1,\mB)= \sbopt(1,\mB)$. 
	Thus, for any such mechanism $M$ it holds for $F'_B$ and $F'_S$ that 
	\[\sbopt(1,\mB)= \opt(1,\mB) \geq \opt(1,1) > M(1,\mB+\lfloor\log_2\mB\rfloor)\]
	as needed. 
	If on the other hand there is no such profile, then clearly the GFT of $M$ is at most the GFT of BTR on the original distributions defined above, and the claim for $M$ also follows, from the claim for $\btr$.
\end{proof}

\subsection{Proof of Proposition~\ref{prop:sd-1-r-ub}}
In this section we restate and prove \cref{prop:sd-1-r-ub}.

\begin{proposition}[\cref{prop:sd-1-r-ub}]
	For any number of buyers $\mB$, any seller distribution $F_S$ and any buyer distribution $F_B$ that FSD $F_S$, it holds that \[\btr(1,\mB+4\sqrt{\mB})\geq\opt(1,\mB).\]
\end{proposition}
\begin{proof}
	 As we will momentarily show in \cref{lem:1-sqrt-ub} below, if $\frac{k(k-1)}{\mB + k} \geq 2$, then  $\opt(1,\mB) \leq \btr(1,\mB+k)$. For $\mB\geq  1$, for any $k\geq 4\sqrt{\mB}$ it holds that $k(k-1) - 2k = k^2-3k \geq  16\mB-12\sqrt{\mB}\geq 2\mB$ and thus 	 
	  $\frac{k(k-1)}{\mB + k} \geq 2$, and the results follows.
\end{proof}

\begin{lemma} \label{lem:1-sqrt-ub}
	If $\frac{k(k-1)}{\mB + k} \geq 2$, then  $\opt(1,\mB) \leq \btr(1,\mB+k)$.
\end{lemma}

\begin{proof} 
We describe an equivalent process to drawing $\mB + k$ buyer values from $F_B$ and 1 seller value from $F_S$, generalizing the process from the proof of \cref{thm:FSD-11-sample}.
First, we draw $\mB+k+1$ quantiles (see \cref{sec:prelim-quantiles}) independently and uniformly from $(0,1)$.  We call these draws, after ordering them, $q^{(1)} \geq q^{(2)}\geq   \ldots \geq  q^{(\mB+k+1)}$.  
For any quantile $q'\in(0,1)$, we denote $b^{q'}=v_{F_B}(q')$
and similarly $s^{q'} = v_{F_S}(q')$.
Since $F_B$ FSD $F_S$, we have that $b^{q'} = v_{F_B}(q') \geq v_{F_S}(q') = s^{q'}$ for any $q'$.

\begin{sloppypar}
After drawing $\vecq$, the sorted vector of quantiles $q^{(1)}, \ldots, q^{(\mB+k+1)}$,
we then draw $i \in \{1, \ldots, \mB+k+1\}$ to be the index of the seller's quantile,
so the seller's value will be $s=s^{q^{(i)}}=v_{F_S}(q^{(i)})$.
The rest of the quantiles will determine the buyer quantiles, that is, we draw $j_1<\ldots<j_{\mB}$ from $\{1, \ldots, \mB+k+1\} \setminus \{i\}$ to be the quantiles of the original buyers, so the values of the original buyers will be $b^{q^{(j_1)}},\ldots,b^{q^{(j_{\mB})}}$. We will let $j^*=j_1$ be the quantile corresponding to the original buyer with the highest value. The remainder of the quantiles belong to the additional buyers; we refer to the set of the indices of the (ordered) quantiles of the additional buyers as $K$.
\end{sloppypar}

As noted in \cref{sec:prelim-quantiles}, generating a value according to a quantile distributed uniformly in $(0,1)$ and then taking the value that corresponds to that quantile for some distribution results in a value distributed according to that distribution.  As such, our procedure is equivalent to drawing $1$ seller value, $\mB$ original buyer values, and $k$ additional buyer values.

As before, we use $b^{(1)} \geq b^{(2)} \geq \cdots\ge b^{(\mB+k)}$ to denote the buyers' values in order.  We sometimes abuse notation by writing $b^{(\ell)}\in K$ to mean that the buyer with the $\ell$\textsuperscript{th} highest value is a new buyer.

In our proof we will reason about the ordering of the seller' and buyers' values along with the ordering of the corresponding quantiles.  For example, the first (highest) quantile may belong to a buyer or a seller.  We might have $b^{(1)} > s$ and yet the first quantile belongs to the seller and the second to the buyer, such that $s = s^{q^{(1)}}$ and $b^{(1)} = b^{q^{(2)}}$ due to the fact that $F_B$ FSD $F_S$.  However, we observe two facts.  First, if the seller has the highest value, then she has the highest quantile.  Second, if a buyer has the highest quantile, then she has the highest value.  We use similar reasoning throughout the proof to move between ordered values, denoted $b^{(\ell)}$, and values ordered quantiles, denoted $b^{q^{(\ell)}}$ (or $s^{q^{(\ell)}}$).

As in the i.i.d.\ case, we compare to the optimal GFT in the augmented market, now with $k$ additional buyers.  The GFT changes in the augmented market only when the buyer with the highest value is a new buyer and has value higher than the seller.  In this case, the difference is the new buyer's gain over the agent who was holding the item post-optimal-trade in the smaller market, i.e., either the highest-value original buyer or the seller, depending on who had the larger value. Therefore,
\begin{align*}
\opt(1, \mB+k) - & \opt(1,\mB)=\\
&= \E_{\vecq,i}\left[(b^{(1)} - \max\{b^{q^{(j^*)}},s\}) \cdot \mathbbm{1}[ b^{(1)} \ge s\ \&\ b^{(1)} \in K \right] \ge\\
& \geq \E_{\vecq,i}\left[(b^{(1)} - \max\{b^{q^{(j^*)}},s\}) \cdot \mathbbm{1}[i > 2\ \&\ 1 \in K] \right] \ge\\
& \geq \E_{\vecq,i}\left[(b^{q^{(1)}} - b^{q^{(\min\{i,j^*\})}}) \cdot \mathbbm{1}[i > 2\ \&\ 1 \in K] \right] \ge\\
&\geq \frac{\mB + k -1}{\mB+k+1} \cdot \frac{k}{\mB + k} \cdot\E_{\vecq}\left[b^{q^{(1)}} - \frac{\mB}{\mB+k-1} b^{q^{(2)}} - \frac{k-1}{\mB + k - 1} b^{q^{(3)}} \right] \ge\\
&\geq \frac{\mB + k -1}{\mB+k+1} \cdot \frac{k}{\mB + k} \cdot\E_{\vecq}\left[\frac{k-1}{\mB + k - 1}b^{q^{(1)}} - \frac{k-1}{\mB + k - 1}b^{q^{(3)}} \right]=\\
&=\frac{\mB + k -1}{\mB+k+1} \cdot \frac{k}{\mB + k} \cdot \frac{k-1}{\mB + k - 1}\cdot\E_{\vecq}\left[b^{q^{(1)}} - b^{q^{(3)}} \right].
\end{align*}

The first inequality follows since $i>2$ implies that $b^{(1)}\ge s$, and $1\in K$ implies that $b^{(1)}\in K$.
The second inequality follows since $i>2$ implies that $b^{(1)} = b^{q^{(1)}}$, and since $s \leq b^{q^{(i)}}$ by FSD hence $\max\{b^{q^{(j^*)}}, s\} \leq \max\{b^{q^{(j^*)}}, b^{q^{(i)}}\} = b^{q^{(\min\{i,j^*\})}}$.

Since $i$ is chosen uniformly at random from $\{1, \ldots, \mB +k+1\}$, then $i > 2$ with probability $\frac{\mB + k -1}{\mB+k+1}$.  Conditioned on this event, 1 is a buyer, and since the original buyers are chosen uniformly from $\{1, \ldots, \mB +k+1\} \setminus \{i\}$, then the probability that $1 \in K$ conditioned on this event is precisely $\frac{k}{\mB + k}$.  Conditioned on the events that $i > 2$ and $1 \in K$, then $\min\{i,j^*\} = 2$ only when $j^* = 2$.  Since $1 \in K$ and $i>2$, this occurs with probability $\frac{\mB}{\mB+k-1}$: the number of original buyers to select over the number of original and remaining additional agents left.  The remaining $\frac{k-1}{\mB + k - 1}$ fraction of the time, $\min\{i,j^*\} \geq 3$.  All of these are independent of $\vecq$. This accounts for the third inequality.
The fourth inequality follows because $b^{q^{(2)}}\le b^{q^{(1)}}$.

Now we compare the optimal GFT in the augmented market to BTR in the augmented market.  Note that BTR only loses when the highest valued agent is a buyer and the second-highest valued agent is a seller, in which case it loses this one (and only) trade.
Therefore,
\begin{align*}
\opt(1,\mB+k)-\btr(1,\mB+k) &= \E_{\vecq,i}\left[(b^{(1)}-s) \cdot \mathbbm{1}[b^{(1)} > s > b^{(2)}] \right] \le\\
&\leq \E_{\vecq,i}\left[(b^{q^{(1)}} - b^{q^{(3)}}) \cdot \mathbbm{1}[b^{(1)} > s > b^{(2)}] \right] \le\\
&\leq \E_{\vecq,i}\left[(b^{q^{(1)}} - b^{q^{(3)}}) \cdot \mathbbm{1}[i \leq 2] \right] \le\\
&\leq \E_{\vecq}\left[b^{q^{(1)}} - b^{q^{(3)}} \right] \cdot \frac{2}{\mB+k+1}.
\end{align*}

We know that $b^{(1)} \le b^{q^{(1)}}$.  When $ s > b^{(2)}$ it must be the case that either $s=s^{q^{(1)}}$ or $s=s^{q^{(2)}}$, and therefore $b^{(2)} = b^{q^{(3)}}$ and $i\le 2$.  The first inequality thus follows since $s > b^{(2)} = b^{q^{(3)}}$, and the second one follows since $ s > b^{(2)}$ implies $i\le2$ (and since $b^{q^{(1)}} - b^{q^{(3)}}$ is always nonnegative).
The third inequality comes from the probability that the seller $i$ is chosen as one of the first two quantiles, which is with probability precisely $\frac{2}{\mB+k+1}$, independently of $\vecq$.

Then for $\opt(1, \mB+k) - \opt(1,\mB) \geq \opt(1,\mB+k)-\btr(1,\mB+k)$ it suffices that
\[\frac{\mB + k -1}{\mB+k+1} \cdot \frac{k}{\mB + k} \cdot \frac{k-1}{\mB + k - 1} \geq \frac{2}{\mB+k+1}, \quad \quad \text{or} \quad \quad \frac{k(k-1)}{\mB + k} \geq 2.\tag*{\qedhere}\]
\end{proof}

\subsection{Proof of Theorem~\ref{btr-1r-converge}}
\label{app:btr-converge}

In this section we prove \cref{btr-1r-converge}. Note that \cref{intro-1r-approx} directly follows from \cref{btr-1r-converge}.
\begin{theorem}[\cref{btr-1r-converge}]\label{btr-1r-approx} 
	For any number of buyers $\mB$, any seller distribution $F_S$, and any buyer distribution $F_B$ such that $F_B~\text{FSD}~F_S$ it holds that 
	\[\btr(1,\mB)\ge\frac{\mB-1}{\mB+1} \cdot \opt(1,\mB).\] 
\end{theorem}
\begin{proof}
	We use similar sampling process and notations as in the proof of  \cref{lem:1-sqrt-ub}.
	First, draw quantiles $q^{(1)}, \ldots, q^{(\mB+1)}$ each i.i.d., uniformly at random from $(0,1)$ and w.l.o.g. assume they are sorted from quantiles corresponding the high values, to quantiles corresponding to low values.  Let $\vecq$ be the vector of the $\mB+1$ quantiles. 
	Then, draw $i \in \{1, \ldots, \mB+1\}$ to be the seller.  
	Note that $i$ denotes the index of the quantile, not the value of the seller. 
	
	The proof is based on the fact that for trade to be reduced it must be the case that $b^{(1)} > s > b^{(2)}$, and since $F_B$ FSD $F_S$ it holds that is such a case, seller $s$ must have the highest or the second-highest quantiles, as otherwise $b^{(2)}\geq s$.
	The expectations in the following are all with respect to the sampling process we have defined above. Starting similarly to the second calculation in proof of  \cref{lem:1-sqrt-ub}, we have that:
	\begin{align*}
	\opt(1,\mB)-\btr(1,\mB) &= \E_{\vecq,i}\left[(b^{(1)}-s) \cdot \mathbbm{1}[b^{(1)} > s > b^{(2)}] \right] \le\\
	&\leq \E_{\vecq,i}\left[(b^{q(1)} - b^{q(3)}) \cdot \mathbbm{1}[b^{(1)} > s > b^{(2)}] \right] \le\\
	&\leq \E_{\vecq,i}\left[(b^{q(1)} - b^{q(3)}) \cdot \mathbbm{1}[i \leq 2] \right] =\\
	&= \E_{\vecq}\left[(b^{q(1)} - b^{q(3)}) \right] \cdot \frac{2}{\mB+1} =\\
	&= \frac{2}{\mB-1}\cdot\E_{q}\left[(b^{q(1)} - b^{q(3)}) \right] \cdot \frac{\mB-1}{\mB+1} =\\
	&= \frac{2}{\mB-1}\cdot\E_{q,i}\left[(b^{q(1)} - b^{q(3)})\cdot \mathbbm{1}[i > 2] \right] \le\\
	&\leq \frac{2}{\mB-1}\cdot\E_{q,i}\left[(b^{(1)}-s)\cdot \mathbbm{1}[i > 2] \right] \le\\
	&\leq \frac{2}{\mB-1}\cdot\btr(1,\mB). \\
	\end{align*}
	Now the claim follows by rearranging, as $1+\frac{2}{\mB-1}= \frac{\mB+1}{\mB-1}$. 
\end{proof}

\section{Missing proofs from Section \ref{sec:mr}}

\subsection{Proof of Theorem~\ref{thm:sd-m-r-ub}}

In this section we restate and prove \cref{thm:sd-m-r-ub}.

\begin{theorem} [\cref{thm:sd-m-r-ub}]
	For any number of sellers $\mS$, any number of buyers $\mB$, any seller distribution $F_S$ and any buyer distribution $F_B$ that FSD $F_S$, it holds that
	\[\btr(\mS,\mS\cdot (\mB+4\sqrt{\mB}))\geq\opt(\mS,\mB).\]
\end{theorem}

As noted in the introduction, \cref{thm:sd-m-r-ub} should indeed first and foremost be viewed as a qualitative result --- that some \emph{finite} number of additional buyers suffices uniformly over all distributions (given stochastic domination) --- and as the first step in quantifying the number of added buyers that is necessary and sufficient to beat the optimum for any pair of distributions (under first-order stochastic dominance). As noted in the introduction, we thus leave the problem of lowering this bound and getting tight quantitative results as our main open problem.

We will prove \cref{thm:sd-m-r-ub} by carefully reducing to the single-seller case of \cref{prop:sd-1-r-ub}. Before we can spell out this reduction, we will need to develop a property of $\opt$ and a property of $\btr$. We believe both of these properties, and especially the latter one (the property of $\btr$), to be also of independent interest.

\begin{lemma}\label{lem:opt-m-1-sum}
	For every $F_B$ and $F_S$ it holds that for every positive integers $m_1,\ldots,m_t$ and $r$:
	\[\opt\left(\sum_{i=1}^t m_i,r\right) \leq  \sum_{i=1}^t \opt(m_i,r).\]
\end{lemma}

\begin{proof}
	We will work in the probability space defined by (independently) drawing $t$ vectors of seller values $\vecs_1,\ldots,\vecs_t$ of respective sizes $m_1,\ldots,m_t$, as well as a vector $\vecb$ of $r$ buyer values. Let $m=\sum_{i=1}^t m_i$. Considering all $m$ sellers, we denote the ordered seller values by $s^{(1)}\le\cdots\le s^{(m)}$. We denote the (reverse-)ordered buyer values by $b^{(1)}\ge\cdots\ge b^{(r)}$. For ease of presentation, if $r<m$ we set $b^{(r+1)}=b^{(r+2)}=\cdots=b^{(m)}=-\infty$. For every $1\le i\le t$ and $1\le k\le m_i$ we let $j(i,k)$ be the index such that the ordered seller values in $\vecs_i$ are $s^{(j(i,1))}\le\cdots\le s^{(j(i,m_i))}$. We observe that $j(i,k)\ge k$ for every $i$ and~$k$. With expectations defined with respect to the above probability space, we then have:
	\begin{multline*}
	\opt\left(\sum_{i=1}^t m_i,r\right)=
	\opt(m,r)=
	\E\left[\sum_{j=1}^m\left(b^{(j)}-s^{(j)}\right)_+\right]=
	\E\left[\sum_{i=1}^t\sum_{k=1}^{m_i}\left(b^{(j(i,k))}-s^{(j(i,k))}\right)_+\right]=\\=
	\sum_{i=1}^t\E\left[\sum_{k=1}^{m_i}\left(b^{(j(i,k))}-s^{(j(i,k))}\right)_+\right]\le
	\sum_{i=1}^t\E\left[\sum_{k=1}^{m_i}\left(b^{(k)}-s^{(j(i,k))}\right)_+\right]=
	\sum_{i=1}^t\opt(m_i,r).\tag*{\qedhere}
	\end{multline*}
\end{proof}

\noindent
From \cref{lem:opt-m-1-sum} we immediately conclude:

\begin{lemma}\label{lem:opt-m-1}
	For every $F_B$ and $F_S$ it holds that for every positive integers $\mB,\mB$:
	\[\opt(\mS,\mB)\leq \mS\cdot \opt(1,\mB).\]
\end{lemma}

\noindent
We next relate BTR in many small markets to BTR in a unified market. 

\begin{lemma}[\cref{lem:btr-m-sum-unified}]\label{lem:btr-m-sum}
	For every $F_B$ and $F_S$ it holds that for every $t>0$ and positive integers $m_1,m_2,\ldots,m_t$ and $r_1,r_2,\ldots,r_t$:
	\[\sum_{i=1}^t \btr(m_i,r_i) \leq  \btr\left(\sum_{i=1}^t m_i,\sum_{i=1}^t r_i\right).\]
\end{lemma}
\begin{proof}
	We will prove this inequality using a pointwise coupling argument: for every $1\le i\le t$ let $\vecs_i$ be a vector of $m_i$ seller values and let $\vecb_i$ be a vector of $r_i$ buyer values. Let $m=\sum_{i=1}^t m_i$ and $r=\sum_{i=1}^t r_i$. Let $\vecs$ be the vector of all $m$ seller values and let $\vecb$ be the vector of all $r$ buyer values. It is enough to show for every such choice of value vectors that:
	\[\sum_{i=1}^t \btr(\vecs_i,\vecb_i)\le\btr(\vecs,\vecb).\]
	We observe that the pre-trade welfare in the ``unified market'' $(\vecs,\vecb)$ equals the \emph{sum} of the pre-trade welfares in the ``small markets'' $(\vecs_i,\vecb_i)$, as both equal $S=\sum_{i=1}^m s^{(i)}$. We will prove the desired inequality by reasoning by cases based on whether a trade is or is not reduced by the $\btr$ mechanism in the unified market.

	We first consider the case in which no trade is reduced by the $\btr$ mechanism in the unified market, that is, the case in which $\btr(\vecs,\vecb)=\opt(\vecs,\vecb)$. In this case, the post-unified-market-optimal-trade \emph{welfare}, i.e., $\opt(\vecb,\vecs)+S$, is the sum of the \emph{$m$ highest values} out of the $m+r$ values in $(\vecs,\vecb)$, while the \emph{sum} of post-small-markets-optimal-trades \emph{welfares}, i.e., $\sum_{i=1}^t\opt(\vecb_i,\vecs_i)+S$, is the sum of \emph{some $m$ values} out of the $m+r$ values in $(\vecs,\vecb)$, and so the former is at least the latter. Therefore, when no trade is reduced in the unified market then:
	\[\sum_{i=1}^t \btr(\vecs_i,\vecb_i)\le\sum_{i=1}^t\opt(\vecs_i,\vecb_i)\le\opt(\vecs,\vecb)=\btr(\vecs,\vecb),\]
	as required.
	
	We now consider the ``interesting'' case, in which some trade is reduced by the $\btr$ mechanism in the unified market.\footnote{The challenge in this case is that it does not necessarily hold that some trade is also reduced by the $\btr$ mechanism in one or more of the small markets.}
	We will first show that in this case, the $m$ highest values out of the $m+r$ values in $(\vecs,\vecb)$ include the value of some buyer $\tilde{b}$ who does not trade in the $\btr$ mechanism in its small market.
	
	Since a trade is reduced by the $\btr$ mechanism in the unified market, the $(m\!+\!1)$st highest value out of the $m+r$ values in $(\vecs,\vecb)$ is a seller value (by \cref{lem:reduce}) --- let us denote this seller by $\tilde{s}$ and let $\tilde{\imath}$ be the index of the small market to which $\tilde{s}$ belongs. We will complete the argument for the existence of $\tilde{b}$ by reasoning by cases based on whether $\tilde{s}$ does or does not trade in small market $\tilde{\imath}$. If $\tilde{s}$ trades in small market $\tilde{\imath}$, then we take $\tilde{b}$ to be the price-setting buyer from small market $\tilde{\imath}$ (which by definition is of higher value than $\tilde{s}$ and does not trade in its small market, $\tilde{\imath}$). If $\tilde{s}$ does not trade in small market $\tilde{\imath}$, then it is one of the $m$ agents who have an item post-small-market-$\btr$-trade; by the pigeonhole principle, therefore one of the $m$ highest values out of the $m+r$ values in $(\vecs,\vecb)$ is of an agent who has no item post-small-market-$\btr$-trade. If this agent is a buyer, then we take it to be $\tilde{b}$ and are done. Otherwise, this agent is a seller who trades in its small market, and we take $\tilde{b}$ to be the price-setting buyer from this agent's market. Either way, we have shown that the $m$ highest values out of the $m+r$ values in $(\vecs,\vecb)$ include the value of some buyer $\tilde{b}$ who does not trade in its small market.
	
	Let $b^-$ be the lowest buyer value of the $m$ highest values out of the $m+r$ values in $(\vecs,\vecb)$, that is, the value of the buyer whose trade with $\tilde{s}$ is reduced by the $\btr$ mechanism in the unified market. We conclude the proof by noting that the post-unified-market-$\btr$-trade \emph{welfare}, i.e., $\btr(\vecb,\vecs)+S$, is the sum of the \emph{$m$ highest values} out of the $m+r-1$ values in $(\vecs,\vecb)\setminus\{b^-\}$, while the sum of post-small-markets-$\btr$-trades welfares, i.e., $\sum_{i=1}^t\btr(\vecb_i,\vecs_i)+S$, is the sum of \emph{some $m$ values} out of the $m+r-1$ values in $(\vecs,\vecb)\setminus\{\tilde{b}\}$. Since by definition $b^-\le\tilde{b}$, we have that the former is at least the latter, and so $\sum_{i=1}^t \btr(\vecs_i,\vecb_i)\le\btr(\vecs,\vecb)$, as required.
\end{proof}

\noindent
From \cref{lem:btr-m-sum} we immediately conclude:

\begin{lemma}\label{lem:btr-m}
	For every $F_B$ and $F_S$ it holds that for every $\mB,\mB$:
	\[\mS\cdot \btr(1,\mB) \leq  \btr(\mS,\mS\cdot \mB).\]
\end{lemma}

\noindent
Using \cref{lem:opt-m-1,lem:btr-m} we can now prove \cref{thm:sd-m-r-ub} by reducing to the single-seller case of \cref{prop:sd-1-r-ub}:

\begin{proof}[Proof of \cref{thm:sd-m-r-ub}]
	By \cref{prop:sd-1-r-ub} it holds that $\btr(1,\mB+4\sqrt{\mB})\geq\opt(1,\mB)$. We combine this with \cref{lem:opt-m-1,lem:btr-m}
	to prove the theorem:
	\[\opt(\mS,\mB)\leq \mS\cdot \opt(1,\mB)\leq \mS\cdot \btr(1,\mB+4\sqrt{\mB}) \leq  \btr(\mS,\mS(\mB+4\sqrt{\mB})).\tag*{\qedhere}\]
\end{proof}

\bibliographystyle{abbrvnat}
\bibliography{refsgft}

\end{document}